\documentclass[aps,prb,twocolumn,reprint,superscriptaddress,floatfix,nofootinbib,longbibliography]{revtex4-2}

\usepackage{graphicx}
\usepackage{amsmath,amssymb,amsthm,mathtools}
\usepackage{bm}
\usepackage{hyperref}
\usepackage{xcolor}
\usepackage{enumitem}
\usepackage{array}
\usepackage{diagbox}

\hypersetup{colorlinks=true,citecolor=blue,linkcolor=blue,urlcolor=blue}

\newtheorem{theorem}{Theorem}
\newtheorem{proposition}[theorem]{Proposition}
\newtheorem{lemma}[theorem]{Lemma}

\newtheorem{conjecture}{Conjecture}
\theoremstyle{remark}
\newtheorem*{remark}{Remark}

\newcommand{\vb}[1]{\bm{#1}}
\newcommand{\hloc}{H^{(2)}}
\newcommand{\len}{\mathrm{len}}

\begin{document}

\title{Absence of nontrivial local conserved quantities in a class of $U(1)$-symmetric spin-1 chains}

\author{Shunsuke Sengoku}
\affiliation{Department of Applied Physics, The University of Tokyo, Tokyo 113-8656, Japan}
\affiliation{Department of Physics, Hong Kong University of Science and Technology, Clear Water Bay, Hong Kong, China}

\author{Haruki Watanabe}
\email{hwatanabe@ust.hk}
\affiliation{Department of Physics, Hong Kong University of Science and Technology, Clear Water Bay, Hong Kong, China}
\affiliation{Center for Theoretical Condensed Matter Physics, Hong Kong University of Science and Technology, Clear Water Bay, Hong Kong, China}
\affiliation{Institute for Advanced Study, Hong Kong University of Science and Technology, Clear Water Bay, Hong Kong, China}

\date{\today}

\begin{abstract}
We prove the absence of nontrivial local conserved quantities in a class of $U(1)$-symmetric spin-$1$ chains with nearest-neighbor interactions in which some of the quadrupolar couplings vanish, a class that is not covered by previous studies. Applying the technique of Shiraishi to these systems, we show that, for every model in this class on a periodic chain of $N$ sites, there is no $k$-local conserved quantity for any $3\le k\le N/2$. In particular, for a frustration-free spin-$1$ chain that exhibits spontaneous $U(1)$ symmetry breaking at zero temperature in one spatial dimension, we prove that every local conserved quantity with support up to half of the system size is a linear combination of the identity, the total magnetization $S^z$, and the Hamiltonian itself. This rigorously establishes that, unlike the Heisenberg ferromagnet, the model admits no local order parameter commuting with the Hamiltonian, so that its continuous symmetry breaking is enabled by the frustration-free structure rather than by a conserved order parameter. We also prove the absence of $k$-local conserved quantities for $3\le k\le N/2$ in the periodic Motzkin chain, a frustration-free spin-$1$ chain closely related to the original Motzkin chain, for which spontaneous $U(1)$ symmetry breaking at zero temperature has also been reported.
\end{abstract}

\maketitle

\section{Introduction}
\label{sec:intro}

An important dichotomy in quantum many-body physics is that between integrable and nonintegrable systems.
Integrable systems admit powerful analytical methods, such as the Bethe ansatz~\cite{Bethe1931}, for the calculation of energy eigenvalues and eigenstates; this solvability is closely tied to the existence of an extensive number of local conserved quantities~\cite{GrabowskiMathieu1995,Kennedy1992YBE,Baxter1982,Korepin1993,CauxMossel2011}. At the same time, such an abundance of local conserved quantities is incompatible with properties expected of generic macroscopic systems, such as thermalization~\cite{Deutsch1991,Srednicki1994,Rigol2008,Rigol2007,Vidmar2016}. It is therefore important to clarify which properties are unique to integrable systems and to identify which systems are nonintegrable. Since the existence or absence of local conserved quantities can be examined directly, it provides a practical criterion for integrability. In this context, Grabowski and Mathieu conjectured that, for Hamiltonians with nearest-neighbor interactions, the existence of a $3$-local conserved quantity is a necessary and sufficient condition for integrability~\cite{GrabowskiMathieu1995Test,GrabowskiMathieu1995}.

Shiraishi~\cite{Shiraishi2019} introduced a systematic technique for
rigorously proving the nonexistence of nontrivial local conserved quantities.
One expands the commutator $[Q,H]$ of a candidate conserved quantity $Q$
with the Hamiltonian $H$ in a basis of operators: if $Q$ is conserved,
the coefficient of every basis operator in $[Q,H]$ must vanish, which
yields a system of linear equations for the expansion coefficients of $Q$.
If these linear equations admit only the trivial solution, then no $Q$ of
the assumed support length can be a conserved quantity. This technique has been applied to spin-$1$ chains including the bilinear-biquadratic (BLBQ) model~\cite{Hokkyo2024,ParkLee2025BLBQ}, to a variety of other quantum spin systems~\cite{Chiba2024mixedIsing,Shiraishi2024NNN,ParkLee2025PXP,Yamaguchi2024classification,Yamaguchi2024proof,ShiraishiTasaki2024,Chiba2025HigherDIsing,fanAbsenceFredkin,Futami2025Compass,Shiraishi2025Classification,Fan2026ThreeSite}, and even to non-spin systems such as the Hubbard and Holstein models~\cite{Futami2025Hubbard,Ishii2026Holstein}. In all of these studies the conjecture of Grabowski and Mathieu has been confirmed, with the exception of non-Hermitian bosonic chains, for which counterexamples were found~\cite{Yamaguchi2026Bosonic}. The technique was further sharpened
by Hokkyo~\cite{Hokkyo2025} into a theorem stating that, under certain
assumptions, the absence of $3$-local quantities whose commutator with the Hamiltonian remains at most $2$-local implies the
absence of $k$-local conserved quantities for all
$3\le k\le N/2$.
We also note a complementary numerical route: the recently proposed symmetry bootstrap detects \emph{a priori unknown} symmetries of a given Hamiltonian from exact-diagonalization data~\cite{Bai2026xSFF}, whereas the analytical technique employed here applies uniformly to families of models with continuous parameters and to arbitrary system sizes.
Throughout this paper, we consider a spin-$1$ chain of $N$ sites under the
periodic boundary condition; the precise definition of $k$-local conserved
quantities is given in Sec.~\ref{sec:notation}.

For spin-$1$ chains, Ref.~\cite{Hokkyo2024} analyzed the general nearest-neighbor Hamiltonian with $U(1)$ symmetry,
\begin{equation}
\begin{aligned}
H = \sum_i\biggl(&\sum_{m\in\{0,\pm1\}}\!\!\!\!e_m E_{m,i}E_{-m,i+1}\\
&+\!\!\!\sum_{m\in\{0,\pm1,\pm2\}}\!\!\!\!\!\!f_m F_{m,i}F_{-m,i+1}+hF_{0,i}\biggr),
\end{aligned}
\label{eq:previous}
\end{equation}
with $e_m\in\mathbb{C}\setminus\{0\}$, $f_m\in \mathbb{C}\setminus\{0\}$, and $h\in\mathbb R$, and established a criterion for the existence of local conserved quantities in models of this form. Here, $E_{m}$ and $F_{m}$ are single-site operators of the same form as the noncommutative analogue of the spherical harmonics~\cite{Madore1992} up to constant factors: the $E_m$ are the rank-$1$ (dipolar) spin operators and the $F_m$ are the rank-$2$ (quadrupolar) operators, so that we refer to the $e_m$ and $f_m$ as the dipolar and quadrupolar couplings, respectively. Their explicit definitions are given in Sec.~\ref{sec:notation}. The proof of Ref.~\cite{Hokkyo2024}, however, relies on all couplings $e_m$ and $f_m$
being nonzero. When some of the $f_m$ vanish, the interaction terms that served as
cancellation partners in the original argument are missing, and the
vanishing-coefficient classes require a genuinely separate treatment.

An important example that motivates such an extension is the following frustration-free chain, whose ground states simultaneously minimize each local term~\cite{TasakiBook,arXiv:2310.16881}:
\begin{equation}
\begin{aligned}
H_i = & -\bigl(S_i^xS_{i+1}^{x}+S_i^yS_{i+1}^y+\Delta\, S_i^z S_{i+1}^z\bigr)\\
      & +\tfrac{1}{\Delta}\bigl(1-(1-\Delta)(S_i^z)^2\bigr)\bigl(1-(1-\Delta)(S_{i+1}^z)^2\bigr),
\end{aligned}
\label{eq:Horig}
\end{equation}
where $H=\sum_iH_i$, $S_i^{x,y,z}$ are the spin-$1$ operators at site $i$, and $\Delta>0$, $\Delta \ne 1$ is a parameter [Eq.~\eqref{eq:Horig} coincides with Eq.~(15) of Ref.~\cite{arXiv:2310.16881} up to the overall normalization $J=\Delta$].
When rewritten in the basis $E_{m},F_{m}$, this model has $f_{\pm1}=f_{\pm2}=0$ [see Eq.~\eqref{eq:coup} below], and is therefore not covered by the analysis of Ref.~\cite{Hokkyo2024}.
This model exhibits spontaneous breaking of the continuous $U(1)$ symmetry at zero temperature in one spatial dimension~\cite{arXiv:2310.16881}, providing an exception to Coleman's theorem~\cite{Coleman1973}, a quantum analogue of the Hohenberg--Mermin--Wagner theorem~\cite{MerminWagner1966,Hohenberg1967}. The exception is closely related to the frustration-free structure: for gapless frustration-free systems, the energy gap of an open-boundary subchain of length $\ell$ is rigorously bounded from above by $O(1/\ell^2)$~\cite{BravyiGosset2015,Gossetozgunov,Anshu,Lemm_2022,lemmCritical2025}, and the analogous $1/L^2$ closing of the gap of the full periodic system of size $L$ has been conjectured~\cite{arXiv:2406.06414}. Correspondingly, gapless excitations of frustration-free systems disperse
quadratically or softer, in contrast to the linearly dispersing
Nambu--Goldstone modes assumed in the standard no-go argument, which is why
one-dimensional symmetry breaking is not excluded~\cite{arXiv:2310.16881,arXiv:2406.06414}. The Heisenberg ferromagnet is the best-known exception to Coleman's theorem, but its mechanism is different: there the order parameter, the magnetization, commutes with the Hamiltonian, and the ordering at $T=0$ is protected by this
commutation---the symmetry-broken ground states are simultaneous eigenstates of
the order parameter and the Hamiltonian, so that quantum fluctuations cannot
destroy the long-range order~\cite{arXiv:2310.16881}. For the model~\eqref{eq:Horig}, in contrast, Ref.~\cite{arXiv:2310.16881} emphasized that the order parameter considered there does \emph{not} commute with the Hamiltonian, so the symmetry breaking is not protected by the mechanism at work in the Heisenberg ferromagnet. Whether some other local order parameter commuting with the Hamiltonian might exist was, however, left open: no proof of the nonexistence was given. Once the nonintegrability proof is extended to cover this model, the nonexistence follows immediately, because any local order parameter commuting with $H$ would be a nontrivial local conserved quantity. Extending the proof technique to spin-$1$ chains with vanishing $f_m$ is thus important not only for the study of integrability itself but also for understanding the mechanism of spontaneous symmetry breaking.

In this work, we extend the nonintegrability proof so that it applies to models in which some of the $f_{m}$ vanish, including the frustration-free chain~\eqref{eq:Horig}. Concretely, we prove that, for all six patterns of vanishing $f_m$ listed in Sec.~\ref{sec:notation}, there is no $k$-local conserved quantity for any $3\le k\le N/2$.
In particular, for the frustration-free chain~\eqref{eq:Horig}, we prove that
every conserved quantity supported on at most $N/2$ consecutive sites is a linear combination of
the trivial ones, namely the identity, the total magnetization $S^z$, and the
Hamiltonian $H$ itself. This establishes the nonexistence of a local order parameter
commuting with the Hamiltonian and thereby indicates that the continuous symmetry breaking of this model is enabled by the frustration-free structure rather than protected by a conserved order parameter. We also prove the absence of $k$-local conserved quantities for
$3\le k\le N/2$ in the periodic Motzkin chain~\cite{Bravyi2012Motzkin,Motzkin2504},
another frustration-free spin-$1$ chain closely related to the original
Motzkin chain, for which spontaneous $U(1)$ symmetry breaking at $T=0$
has also been reported, diagnosed through the long-range order of transverse spin correlations in its unique ground state~\cite{Menon2024} (Appendix~\ref{app:motzkin}).

This paper is organized as follows.
In Sec.~\ref{sec:notation}, we introduce the Hamiltonian, the operator basis following Ref.~\cite{Hokkyo2024}, and the notation, and we outline the proof.  Section~\ref{sec:k2} shows that the strictly $2$-local part of any 2-local conserved quantity must be proportional to the interaction part $\sum_i\hloc_i$ of the Hamiltonian, which implies that every $2$-local conserved quantity is, up
to $1$-local terms and the identity, a multiple of the Hamiltonian itself.  Section~\ref{sec:reduction} invokes the theorems of Ref.~\cite{Hokkyo2025} to reduce the analysis of $k$-local conserved quantities with $3 \le k \le N/2$ to that of
$3$-local quantities whose commutator with $H$ is at most $2$-local.  Section~\ref{sec:Ceq} presents the necessary conditions $(C\text{-}1)$--$(C\text{-}10)$ on the coupling constants for the existence of such a quantity.  Section~\ref{sec:3local} shows that these conditions cannot be satisfied in any of the six patterns, thereby ruling out $k$-local conserved quantities for all $3\le k\le N/2$.  Section~\ref{sec:k1} analyzes
$1$-local conserved quantities, which completes the argument for the chain~\eqref{eq:Horig}.  Section~\ref{sec:conclusion} is devoted to conclusions.  Appendix~\ref{app:Ceq} derives the
constraint equations $(C\text{-}1)$--$(C\text{-}10)$, Appendix~\ref{app:Cfull} collects their full list, Appendix~\ref{app:motzkin} analyzes the periodic Motzkin
chain~\cite{Bravyi2012Motzkin,Motzkin2504}, and Appendix~\ref{app:motzkin-gs} proves the frustration-freeness and the ground-state degeneracy of the periodic Motzkin chain.

\section{Model and proof outline}
\label{sec:notation}

\subsection{Model}

We follow the notation of Ref.~\cite{Hokkyo2024}, recalled briefly here.
The space of operators acting on a single spin-$1$ site is nine dimensional.  As its basis we use the set~\cite{Hokkyo2024}
$\{I,E_0,E_{\pm1},F_0,F_{\pm1},F_{\pm2}\}$, where $I$ is the identity and the other eight elements are traceless: the $E_m$ are the rank-$1$ (dipolar) and the $F_m$ the rank-$2$ (quadrupolar) operators, defined explicitly by
\begin{equation}
\begin{aligned}
E_0 &= \begin{pmatrix}1&0&0\\0&0&0\\0&0&-1\end{pmatrix} = S^z,\\
E_{+1} &= \begin{pmatrix}0&1&0\\0&0&1\\0&0&0\end{pmatrix} = \tfrac{S^x+iS^y}{\sqrt2},\\
F_{0} &= \begin{pmatrix}1&0&0\\0&-2&0\\0&0&1\end{pmatrix}\\
       &\;\;= -(S^x)^2-(S^y)^2+2(S^z)^2,\\
F_{+1} &= \begin{pmatrix}0&1&0\\0&0&-1\\0&0&0\end{pmatrix} = \bigl\{S^z,\tfrac{S^x+iS^y}{\sqrt2}\bigr\},\\
F_{+2} &= \begin{pmatrix}0&0&1\\0&0&0\\0&0&0\end{pmatrix} = \bigl(\tfrac{S^x+iS^y}{\sqrt2}\bigr)^2,\\
E_{-m} &= E_{+m}^\dagger,\quad F_{-m} = F_{+m}^\dagger.
\end{aligned}
\label{eq:basis}
\end{equation}
Here, the matrices are written in the basis $\{|{+}1\rangle,|0\rangle,|{-}1\rangle\}$ of eigenstates of $S^z$, and $\{a,b\}:=ab+ba$ denotes the anticommutator.
The operator basis~\eqref{eq:basis} closes under both commutation and
anticommutation: for any two basis elements $a,b$, the commutator $[a,b]$
and the anticommutator $\{a,b\}$ are again linear combinations of the basis
elements, and follow the relations listed in Table~\ref{tab:comm} and
Table~\ref{tab:anticomm}, respectively. In particular, every commutator of
two basis elements is proportional to a single basis element.

It is useful to assign to each basis element $O$ the \emph{weight} $w(O)$ defined by the commutation relation with $E_0=S^z$,
\begin{equation}
[E_0,O]=w(O)\,O,\quad
w(O)=
\begin{cases}
0 & (O=I,E_0,F_0),\\
\pm1 & (O=E_{\pm1},F_{\pm1}),\\
\pm2 & (O=F_{\pm2}).
\end{cases}
\label{eq:weight-def}
\end{equation}
The weight is additive under operator products, and the total weight of a tensor product of basis elements is the sum of the weights of its factors.

\begin{table*}[t!]
\centering
\def\ep{{E_{+1}}}\def\ez{{E_0}}\def\em{{E_{-1}}}\def\fpp{{F_{+2}}}\def\fp{{F_{+1}}}\def\fz{{F_0}}\def\fm{{F_{-1}}}\def\fmm{{F_{-2}}}%
 \caption{\label{tab:comm} The commutators $[a, b]$ where $a$ and $b$ are elements of our operator basis~\eqref{eq:basis}.
  Each of these commutators is proportional to a single element of this operator basis.
 }
\def\arraystretch{1.5}
\begin{tabular}{|c||c|c|c|c|c|c|c|c|} \hline
   \diagbox[dir=NW]{a}{b} & $\ep$ & $\ez$ & $\em$ & $\fpp$ & $\fp$ & $\fz$ & $\fm$ & $\fmm$ \\ \hline\hline
    $\ep$&$0$&$-\ep$&$+\ez$&$0$&$-2\fpp$&$-3\fp$&$+\fz$&$+\fm$\\ \hline
    $\ez$&$+\ep$&$0$&$-\em$&$+2\fpp$&$+\fp$&$0$&$-\fm$&$-2\fmm$\\ \hline
    $\em$&$-\ez$&$+\em$&$0$&$-\fp$&$-\fz$&$+3\fm$&$+2\fmm$&$0$\\ \hline
    $\fpp$&$0$&$-2\fpp$&$+\fp$&$0$&$0$&$0$&$-\ep$&$+\ez$\\\hline
    $\fp$&$+2\fpp$&$-\fp$&$+\fz$&$0$&$0$&$-3\ep$&$+\ez$&$-\em$\\\hline
    $\fz$&$+3\fp$&$0$&$-3\fm$&$0$&$+3\ep$&$0$&$-3\em$&$0$\\\hline
    $\fm$&$-\fz$&$+\fm$&$-2\fmm$&$+\ep$&$-\ez$&$+3\em$&$0$&$0$\\\hline
    $\fmm$&$-\fm$&$+2\fmm$&$0$&$-\ez$&$+\em$&$0$&$0$&$0$\\\hline
  \end{tabular}
\end{table*}

\begin{table*}[t!]
\centering
\def\ep{{E_{+1}}}\def\ez{{E_0}}\def\em{{E_{-1}}}\def\fpp{{F_{+2}}}\def\fp{{F_{+1}}}\def\fz{{F_0}}\def\fm{{F_{-1}}}\def\fmm{{F_{-2}}}%
\caption{\label{tab:anticomm} The anticommutator $\{a, b\}$ of elements $a$ and $b$ in the operator basis.}
  \def\arraystretch{1.5}
\begin{tabular}{|c||c|c|c|c|c|c|c|c|} \hline
   \diagbox[dir=NW]{a}{b} & $\ep$ & $\ez$ & $\em$ & $\fpp$ & $\fp$ & $\fz$ & $\fm$ & $\fmm$ \\ \hline\hline
    $\ep$ & $2\fpp$ & $\fp$ & $\frac{4I-\fz}{3}$ & $0$ &  $0$ & $-\ep$ &$\ez$ &$\em$ \\ \hline
    $\ez$ & $\fp$& $\frac{4I+2\fz}{3}$ & $\fm$ & $0$ & $\ep$ & $2\ez$ & $\em$ & $0$\\ \hline
    $\em$ & $\frac{4I-\fz}{3}$ & $\fm$ & $2\fmm$ & $\ep$ & $\ez$ & $-\em$ & $0$ & $0$ \\ \hline
    $\fpp$ & $0$ & $0$ & $\ep$ & $0$ & $0$ & $2\fpp$ & $-\fp$ & $\frac{2I+\fz}{3}$ \\\hline
    $\fp$ & $0$ & $\ep$ & $\ez$ & $0$ & $-2\fpp$ & $-\fp$ & $\frac{4I-\fz}{3}$ & $-\fm$ \\\hline
    $\fz$ & $-\ep$ & $2\ez$ & $-\em$ & $2\fpp$ & $-\fp$ & $4I-2\fz$ & $-\fm$ & $2\fmm$ \\\hline
    $\fm$ & $\ez$ & $\em$ & $0$ & $-\fp$ & $\frac{4I-\fz}{3}$ & $-\fm$ & $-2\fmm$ & $0$ \\\hline
    $\fmm$ & $\em$ & $0$ & $0$ & $\frac{2I+\fz}{3}$ & $-\fm$ & $2\fmm$ & $0$ & $0$ \\\hline
  \end{tabular}
\end{table*}

In this work, we consider the following Hamiltonian on a periodic chain of $N$ sites (site $N+1$ is identified with site $1$):
\begin{equation}
\begin{aligned}
H = \sum_{i=1}^N\biggl(&\sum_{m\in\{0,\pm1\}}\!\!\!\!e_m E_{m,i}E_{-m,i+1}\\
&+\!\!\!\sum_{m\in\{0,\pm1,\pm2\}}\!\!\!\!\!\!f_m F_{m,i}F_{-m,i+1}+hF_{0,i}\biggr),
\end{aligned}
\label{eq:Hgen}
\end{equation}
with $e_m\in\mathbb{C}\setminus\{0\}$, $f_m\in \mathbb{C}$, and $h\in\mathbb R$, where $A_{i}$ denotes the basis element $A$ acting on site $i$.
Hermiticity of $H$ enforces $e_{-m}=e_m^\ast$ and $f_{-m}=f_m^\ast$; in particular, $e_0$ and $f_0$ are real.
It is convenient to split the Hamiltonian into the nearest-neighbor interaction part and the on-site part as
\begin{equation}
H=\sum_{i=1}^N\bigl(\hloc_{i}+H^{(1)}_i\bigr),
\label{eq:H-decomp}
\end{equation}
with
\begin{equation}
\begin{aligned}
\hloc_{i}&:=\sum_{m\in\{0,\pm1\}}\!\!e_m E_{m,i}E_{-m,i+1}\\
&\hphantom{:=}+\!\!\sum_{m\in\{0,\pm1,\pm2\}}\!\!\!f_m F_{m,i}F_{-m,i+1},\\
H^{(1)}_i&:=hF_{0,i},
\end{aligned}
\label{eq:H2H1-def}
\end{equation}
where $\hloc_{i}$ acts on the two sites $i,i+1$ and $H^{(1)}_i$ acts on site $i$.

Since every term of $H$ has total weight zero, the Hamiltonian~\eqref{eq:Hgen} has a $U(1)$ symmetry generated by the total magnetization $S^z:=\sum_{i=1}^N E_{0,i}$, i.e., $[S^z,H]=0$. Hence the identity, $S^z$, and $H$ itself, as well as their linear combinations, are always conserved; we refer to these as the \emph{trivial} conserved quantities.
The Hamiltonian~\eqref{eq:Horig} can be rewritten, up to the additive constant $N(1+2\Delta)^2/(9\Delta)$ (a multiple of the identity, which does not affect the analysis of conserved quantities), as
\begin{equation}
H = \sum_{i}\!\Bigl(\!\sum_{m=0,\pm1}\!\!e_{m}E_{m,i}E_{-m,i+1}+f_0F_{0,i}F_{0,i+1}+hF_{0,i}\Bigr),
\label{eq:H}
\end{equation}
with
\begin{equation}
\begin{aligned}
e_0 &= -\Delta, & e_{\pm1} &= -1,\\
f_0 &= \tfrac{(1-\Delta)^2}{9\Delta}, & h &= -\tfrac{2(1+2\Delta)(1-\Delta)}{9\Delta}.
\end{aligned}
\label{eq:coup}
\end{equation}
Thus, the frustration-free chain~\eqref{eq:Horig} corresponds to the special case of Eq.~\eqref{eq:Hgen} with
$f_{\pm1}=f_{\pm2}=0$.

The case in which all of the quadrupolar couplings $f_0,f_{\pm1},f_{\pm2}$ are nonzero was analyzed in Ref.~\cite{Hokkyo2024}. Note that, by the Hermiticity constraint $f_{-m}=f_m^\ast$, the coefficient $f_{+m}$ vanishes if and only if $f_{-m}$ does. In this work we consider the remaining cases in which some of the $f_m$ vanish, which fall into the following six patterns:
\begin{enumerate}[label=\textup{(\roman*)},itemsep=2pt]
\item $f_{\pm1}=f_{\pm2}=0$, $f_0\ne0$;
\item $f_{\pm1}=0$, $f_0, f_{\pm2}\ne0$;
\item $f_0=0$, $f_{\pm1},f_{\pm2}\ne0$;
\item $f_{\pm2}=0$, $f_0,f_{\pm1}\ne0$;
\item $f_{\pm1}=f_0=0$, $f_{\pm2}\ne0$;
\item $f_{\pm2}=f_0=0$, $f_{\pm1}\ne0$.
\end{enumerate}
The frustration-free chain~\eqref{eq:Horig} belongs to pattern (i).
The only pattern not covered by this work or by Ref.~\cite{Hokkyo2024} is the one with $f_0=f_{\pm1}=f_{\pm2}=0$, i.e., the spin-$1$ XXZ chain with a single-ion anisotropy but without biquadratic exchange, which requires a separate analysis and is left for future work.

Our main result is the following theorem (the precise definition of $k$-locality is given in Sec.~\ref{ssec:column} below).
\begin{theorem}[Main result]\label{thm:main}
Consider the Hamiltonian~\eqref{eq:Hgen} on a periodic chain of $N$ sites with all $e_m\neq0$, belonging to any of the patterns \textup{(i)}--\textup{(vi)}. Then, for every $3\le k\le N/2$, there is no $k$-local conserved quantity. Consequently, every conserved quantity $Q$ with $\len(Q)\le N/2$ is a linear combination of the identity, $1$-local conserved quantities, and $H$ itself. Moreover, for the frustration-free chain~\eqref{eq:Horig} with $\Delta>0$ and $\Delta\neq1$ [pattern \textup{(i)} with the couplings~\eqref{eq:coup}], every conserved quantity $Q$ with $\len(Q)\le N/2$ is a linear combination of the identity, $S^z$, and $H$.
\end{theorem}

\begin{remark}
The upper bound $k\le N/2$ in Theorem~\ref{thm:main} is optimal: conserved quantities beyond the trivial linear combinations exist as soon as $k$ exceeds $N/2$.
The simplest example is $(S^z)^2$, which contains the basis strings $E_{0,i}E_{0,i+\lfloor N/2\rfloor}$ (with identities on the intermediate sites) with coefficient $2$ and hence satisfies $\len\bigl((S^z)^2\bigr)=\lfloor N/2\rfloor+1$.
More generally, the longest basis strings arise from placing the constituent single-site factors or bond terms at equal spacing on the ring, which gives
$\len\bigl((S^z)^n\bigr)=N-\lceil N/n\rceil+1$ for $2\le n\le N$ and
$\len\bigl(H^n\bigr)=N-\lceil N/n\rceil+2$ for $2\le n\le N/2$;
powers of the trivial conserved quantities thus furnish conserved quantities at values of $k$ accumulating at $k=N$.
For $k>N/2$, the meaningful question is therefore not the nonexistence of nontrivial conserved quantities but their classification modulo polynomials in $S^z$ and $H$~\cite{Chiba2024mixedIsing}.
Note also that a $k$-local conserved quantity with $k\le N/2$ remains a conserved quantity when the same translation-invariant Hamiltonian is placed on a larger chain, whereas quantities such as $H^2$ are specific to the given system size~\cite{Chiba2024mixedIsing,Hokkyo2024}.
\end{remark}

\subsection{Notation}
\label{ssec:column}
We call a tensor product of basis elements on $l$ consecutive sites whose leftmost and rightmost factors are not the identity, such as
$\vb A_i^l=A_{(1),i}A_{(2),i+1}\cdots A_{(l),i+l-1}$ with $A_{(1)},A_{(l)}\ne I$,
a \emph{basis string} of length $l$ with leftmost site $i$ [for $l=1$, $\vb A_i^1=A_{(1),i}$ with $A_{(1)}\neq I$].
Any operator $Q$ can be expanded uniquely in terms of basis strings as
\begin{equation}
Q=q_0 I+\sum_{l\ge1}\sum_{i}\sum_{\vb A_i^l}q_{\vb A_i^l}\,\vb A_i^l,
\label{eq:Q-expand}
\end{equation}
with complex coefficients $q_{\vb A_i^l}$.
We define $\len(Q)$ as the largest $l$ such that $q_{\vb A_i^l}\neq0$ for some basis string of length $l$ [and $\len(Q)=0$ if $Q\propto I$].
A \emph{$k$-local conserved quantity} is an operator $Q$ with $\len(Q)=k$ that commutes with $H$.
Note that with this definition a $k$-local conserved quantity may also contain basis strings shorter than $k$.
For example, $H$ itself is a $2$-local conserved quantity, and $S^z$ is a $1$-local conserved quantity.
We often write $q(A_{(1)}A_{(2)}\cdots)_i$ for the coefficient of the basis string $A_{(1),i}A_{(2),i+1}\cdots$ in $Q$, and omit the site index $i$ when it is not important.

A convenient technical device introduced in Ref.~\cite{Shiraishi2019} is a
compact graphical notation for commutators of basis strings. Consider for instance
\begin{equation}
\begin{aligned}
&[E_{0,i}E_{0,i+1}E_{+1,i+2},\,E_{-1,i+2}E_{+1,i+3}]\\
&\quad=E_{0,i}E_{0,i+1}E_{0,i+2}E_{+1,i+3},
\end{aligned}
\label{eq:column-example}
\end{equation}
which follows from $[E_{+1},E_{-1}]=E_0$.  We represent this
calculation by the column expression
\begin{equation}
\begin{array}{ccccc}
&E_0 & E_0 & E_{+1} & \\
    &&     & E_{-1} & E_{+1}\\
\hline
+& E_0 & E_0 & E_0 & E_{+1}
\end{array}.
\label{eq:column-canonical}
\end{equation}
In this notation, the first row shows a basis string appearing in the candidate conserved quantity $Q$, the second row shows a term of the Hamiltonian, and the row below the
horizontal line is the resulting commutator (first row commuted with second row) of these two operators. The columns are aligned with the lattice
sites on which the operators act, so that the relative position of the two operators is encoded by the horizontal offset of
the second row. Throughout the rest of the paper
we use this notation freely; the site index $i$ is implicit.

\subsection{Outline of the proof}
\label{ssec:outline}

The strategy of Refs.~\cite{Shiraishi2019,Hokkyo2024} for
proving the absence of nontrivial local conserved quantities is as follows.
Let $Q$ be a $k$-local quantity expanded as in Eq.~\eqref{eq:Q-expand}.
Since $H$ contains only $1$-local and $2$-local terms, the commutator $[Q,H]$ is at
most $(k+1)$-local:
\begin{equation}
[Q,H]
=\sum_{l=0}^{k+1}\sum_{i}\sum_{\vb B_i^l}r_{\vb B_i^l}\,\vb B_i^l,
\label{eq:QH-expand}
\end{equation}
where each coefficient $r_{\vb B_i^l}$ is a linear combination of the coefficients
$\{q_{\vb A_i^l}\}$ determined by the commutation relations of
Table~\ref{tab:comm} and the anticommutation relations of Table~\ref{tab:anticomm}. The conservation condition
$[Q,H]=0$ is then equivalent to the linear system
\begin{equation}
r_{\vb B_i^l}=0\quad\text{for all }\vb B_i^l\;\text{and}\;l\le k+1,
\label{eq:r-zero}
\end{equation}
for the coefficients $\{q_{\vb A_i^l}\}$. The goal of the
proof is to show that every solution of~\eqref{eq:r-zero} has vanishing length-$k$ coefficients, so that no $k$-local conserved quantity exists.

Our proof of Theorem~\ref{thm:main} proceeds in the following steps.
In Sec.~\ref{sec:k2}, we first show that, in any of the patterns (i)--(vi), a strictly $2$-local quantity $X$ whose commutator $[X,H]$ contains no length-$3$ terms must be proportional to $\sum_i \hloc_i$ (Proposition~\ref{prop:doubling}); in particular, every $2$-local conserved quantity is a linear combination of $H$, $1$-local conserved quantities, and the identity.
This property allows us to apply the theorems of Ref.~\cite{Hokkyo2025}, reviewed in Sec.~\ref{sec:reduction}, which reduce the analysis of $k$-local conserved quantities for all $3\le k\le N/2$ to that of $3$-local quantities whose commutator with $H$ is at most $2$-local and, moreover, fix the strictly $3$-local part of any candidate to a specific ``doubling'' form.
In Sec.~\ref{sec:Ceq} (with derivations in Appendix~\ref{app:Ceq}), we recall the ten necessary conditions $(C\text{-}1)$--$(C\text{-}10)$ on the coupling constants for the existence of a $3$-local quantity whose commutator with $H$ is at most $2$-local.
In Sec.~\ref{sec:3local}, we show that these conditions can never be satisfied within the patterns (i)--(vi) when all $e_m\ne0$, which excludes $k$-local conserved quantities for all $3\le k\le N/2$.
Finally, in Sec.~\ref{sec:k1}, we determine the $1$-local conserved quantities; for the frustration-free chain~\eqref{eq:Horig} the only one is $S^z$ (up to normalization), which completes the proof of Theorem~\ref{thm:main}.
\section{$2$-local conserved quantities}
\label{sec:k2}

In this section, we characterize the strictly $2$-local quantities $X$ whose commutator $[X,H]$ contains no length-$3$ basis strings: they are exhausted by the scalar multiples of the interaction part $\sum_i\hloc_i$ (Proposition~\ref{prop:doubling} below).  In particular, every $2$-local conserved quantity is, up
to $1$-local conserved quantities and the identity, a scalar multiple of the Hamiltonian itself.  This
verifies that our systems meet the assumptions of the reduction theorems of Ref.~\cite{Hokkyo2025}, which we use in Sec.~\ref{sec:reduction} to rule out $k$-local conserved quantities for all $3\le k\le N/2$.

Following Ref.~\cite{Hokkyo2024}, we call a length-$2$ basis string of the form $C_iC^\dagger_{i+1}$, with $C\in\{E_0,E_{\pm1},F_0,F_{\pm1},F_{\pm2}\}$, a \emph{doubling operator} (in Ref.~\cite{Hokkyo2024} the same term also refers to the analogous stacked strings of general length $k$; in this paper we reserve it for the length-$2$ case).  Note that the interaction part $\hloc_i$ of Eq.~\eqref{eq:H2H1-def} is a linear combination of doubling operators: $\hloc_i=\sum_{C}c_C\,C_iC^\dagger_{i+1}$ with $c_{E_m}=e_m$ and $c_{F_m}=f_m$.  The main result of this section is the following.

\begin{proposition}\label{prop:doubling}
Consider the Hamiltonian~\eqref{eq:Hgen} with any of the patterns \textup{(i)}--\textup{(vi)} of Sec.~\ref{sec:notation}, and let
$X=\sum_i\sum_{\vb A_i^2}q_{\vb A_i^2}\vb A_i^2$ be a strictly $2$-local quantity such that $[X,H]$ contains no length-$3$ basis strings.
Then the only basis strings $\vb A_i^2=A_{(1),i}A_{(2),i+1}$ that may have
a nonvanishing coefficient in $X$ are the
doubling operators $C_iC^\dagger_{i+1}$ appearing in $\hloc$ (i.e., those with $c_C\neq0$),
and their coefficients satisfy
\begin{equation}
q(CC^\dagger)_i = q_{k=2}\,c_C,
\label{eq:k2-coef}
\end{equation}
with a constant $q_{k=2}$
independent of the site $i$ and of $C$; that is, $X=q_{k=2}\sum_i \hloc_i$.
\end{proposition}

Before proving the proposition, we note its consequence for $2$-local conserved quantities.  Let $Q$ be a $2$-local conserved quantity and let $X$ be its strictly $2$-local part.  Since the length-$3$ part of $[Q,H]$, which must vanish, receives contributions only from $X$, the operator $X$ satisfies the assumption of Proposition~\ref{prop:doubling}, and hence $X=q_{k=2}\sum_i\hloc_i$.  Using $\sum_i\hloc_i=H-h\sum_iF_{0,i}$, we find that
\begin{equation}
Q-q_{k=2}H=\hat X^{(1)}+q_{k=2}h\sum_iF_{0,i}+\text{const},
\end{equation}
where $\hat X^{(1)}$ denotes the strictly $1$-local part of $Q$, is an operator with $\len\le1$, which is itself conserved because $Q$ and $H$ are.  Therefore, every $2$-local conserved quantity is a linear combination of $H$, $1$-local conserved quantities, and the identity.

We now prove Proposition~\ref{prop:doubling}, first for pattern (i) and then for patterns (ii)--(vi).

\subsection{Pattern (i): $f_{\pm1}=f_{\pm2}=0$}
\label{ssec:k2-pi}

Let $X$ be a strictly $2$-local quantity as in Proposition~\ref{prop:doubling} and write the coefficient of
$\vb A_i^2=A_{(1)}A_{(2)}$ as $q(A_{(1)}A_{(2)})_i$. We first show that
$q(A_{(1)}A_{(2)})\ne 0$ forces $A_{(2)}=A_{(1)}^\dagger$, in the following four steps.

\subsubsection*{Step 1: $A_{(j)}\in\{F_{\pm1},F_{\pm2}\}$ gives zero
coefficient.} For instance, take $A_{(1)}=F_{+1}$ and
$A_{(2)}=E_{-1}$. The only generating pair for the length-$3$ basis string
$F_{+1}E_{-1}E_0$ in $[X,H]$ is
\begin{equation}
\begin{array}{ccc}
F_{+1} & E_{-1} & \\
       & E_0 & E_0 \\
\hline
F_{+1} & E_{-1} & E_0
\end{array}.
\end{equation}
Indeed, the alternative placement, in which the Hamiltonian term sits on the left pair of sites, would require an interaction term of $H$ containing $F_{+1}$, which is absent in pattern (i).
Hence the coefficient of $F_{+1}E_{-1}E_0$ in $[X,H]$ is proportional to $e_0\,q(F_{+1}E_{-1})_i$ with no second pair available to cancel it, so
$q(F_{+1}E_{-1})=0$. Similar arguments give zero coefficients whenever
$A_{(1)}$ or $A_{(2)}$ belongs to $\{F_{\pm 1},F_{\pm 2}\}$. We henceforth assume
$A_{(1)},A_{(2)}\notin\{F_{\pm 1},F_{\pm 2}\}$.

\subsubsection*{Step 2: $A_{(2)}=E_0$ forces $A_{(1)}=E_0$.}
Consider the length-$3$ basis string $A_{(1)}E_{+1}E_{-1}$. It is generated by the two placements
\begin{equation}
\begin{array}{ccc}
A_{(1)} & E_0 & \\
        & E_{+1} & E_{-1} \\
\hline
A_{(1)} & E_{+1} & E_{-1}
\end{array},\;\;
\begin{array}{cccc}
& & D & E_{-1} \\
& A_{(1)} & A_{(1)}^\dagger & \\
\hline
\propto & A_{(1)} & E_{+1} & E_{-1}
\end{array},
\end{equation}
where $D$ must satisfy $[A_{(1)}^\dagger,D]\propto E_{+1}$.
For the coefficient of $A_{(1)}E_{+1}E_{-1}$ in $[X,H]$ to vanish, the first contribution, proportional to $q(A_{(1)}E_0)_i$, must be cancelled by contributions of the second type, proportional to $q(DE_{-1})_{i+1}$.
Among the pairs with $A_{(1)},D\notin\{F_{\pm1},F_{\pm2}\}$ (Step 1 removes the others), the condition $[A_{(1)}^\dagger,D]\propto E_{+1}$ is satisfied
only for $(D,A_{(1)})=(E_{+1},E_0)$ and $(E_0,E_{-1})$. By the analogous consideration of the length-$3$ basis string $A_{(1)}E_{-1}E_{+1}$, the pairs satisfying
$[A_{(1)}^\dagger,D]\propto E_{-1}$ are $(D,A_{(1)})=(E_{-1},E_0)$ and $(E_0,E_{+1})$.
A nonzero $q(A_{(1)}E_0)$ requires a cancellation partner in both cases, and the only $A_{(1)}$ appearing in both lists is $A_{(1)}=E_0$; hence $q(A_{(1)}E_0)=0$ unless $A_{(1)}=E_0$.

\subsubsection*{Step 3: $A_{(2)}=F_0$ forces $A_{(1)}=F_0$.}
This follows in the same way by enumerating the placements generating the length-$3$ basis strings
$A_{(1)} F_{+1} E_{-1}$ and $A_{(1)} F_{-1} E_{+1}$: the contribution proportional to $q(A_{(1)}F_0)_i$ [via $[F_0,E_{\pm1}]=\pm3F_{\pm1}$] must be cancelled by a contribution proportional to some $q(DE_{\mp1})_{i+1}$ with $[A_{(1)}^\dagger,D]\propto F_{\pm1}$, whose coefficient survives Steps~1 and~2 only if $A_{(1)}=F_0$ (with $D=F_{\pm1}$ excluded by Step~1).

\subsubsection*{Step 4: $A_{(2)}\notin\{E_0,F_0\}$ forces $A_{(1)}=A_{(2)}^\dagger$.}
Since $A_{(2)}\notin\{I,E_0,F_0\}$ has nonzero weight, $[A_{(2)},E_0]=-w(A_{(2)})A_{(2)}\neq0$, and the length-$3$ basis string $A_{(1)}A_{(2)}E_0$ is generated by the placement of $e_0E_0E_0\in H$ on the right pair of sites, with coefficient proportional to $q(A_{(1)}A_{(2)})_i$.
The only possible cancellation partners are of the form $q(DE_0)_{i+1}$ with $[A_{(1)}^\dagger,D]\propto A_{(2)}$, and by Step 2 we need $D=E_0$. The condition
$[A_{(1)}^\dagger,E_0]\propto A_{(2)}$ with a nonvanishing proportionality constant then gives
$A_{(1)}=A_{(2)}^\dagger$ whenever $q(A_{(1)}A_{(2)})_i\neq0$.

It remains to determine the coefficients of the doubling operators.
For two doubling operators $C_{(1)}C_{(1)}^\dagger$ and $C_{(2)}C_{(2)}^\dagger$ with $[C_{(1)}^\dagger,C_{(2)}]\neq0$,
the length-$3$ basis string
$\vb B^3_i \propto C_{(1)}\,[C_{(1)}^\dagger,C_{(2)}]\,C_{(2)}^\dagger$
in $[X,H]$ is generated by exactly two placements:
\begin{equation}
\begin{array}{ccc}
C_{(1)} & C_{(1)}^\dagger & \\
        & C_{(2)} & C_{(2)}^\dagger \\
\hline
C_{(1)} & [C_{(1)}^\dagger,C_{(2)}] & C_{(2)}^\dagger
\end{array},\;\;
\begin{array}{cccc}
&&C_{(2)} & C_{(2)}^\dagger \\
&C_{(1)} & C_{(1)}^\dagger & \\
\hline
- & C_{(1)} & [C_{(1)}^\dagger,C_{(2)}] & C_{(2)}^\dagger
\end{array},
\end{equation}
whose cancellation yields
\begin{equation}
c_{C_{(2)}}\,q(C_{(1)}C_{(1)}^\dagger)_i
=c_{C_{(1)}}\,q(C_{(2)}C_{(2)}^\dagger)_{i+1}.
\label{eq:k2-ratio}
\end{equation}
Taking $C_{(1)}=C_{(2)}=C$ with $[C^\dagger,C]\ne0$ (e.g., $C=E_{\pm1}$) in Eq.~\eqref{eq:k2-ratio} shows that $q(CC^\dagger)_i$ is independent of $i$; taking $C_{(1)}\neq C_{(2)}$ then shows that the ratio $q(CC^\dagger)/c_{C}$ is common to all doubling operators appearing in $\hloc$ [in pattern (i), the relevant pairs $(C_{(1)},C_{(2)})$, such as $(E_0,E_{+1})$, $(E_{+1},E_{-1})$, and $(F_0,E_{\pm1})$, connect all of $C\in\{E_0,E_{\pm1},F_0\}$].
This is precisely~\eqref{eq:k2-coef}, and completes the proof for pattern (i).

\subsection{Patterns (ii)--(vi)}
\label{ssec:k2-rest}
The proofs for the remaining patterns are similar to that of pattern (i) and to the analysis of Ref.~\cite{Hokkyo2024}, except for pattern (iv), which requires a separate argument. Below, we outline the proof for pattern (ii) briefly, comment on patterns (iii), (v), and (vi), and then give the proof for pattern (iv) in detail.
\subsubsection{Pattern (ii): $f_{\pm 1}=0$, $f_0,f_{\pm 2}\ne 0$}
Writing $\vb A_i = A_{(1)} A_{(2)}$, if $A_{(1)}=F_{\pm 1}$ or $A_{(2)}=F_{\pm 1}$,
the coefficient vanishes for the same reason as in Step 1 of pattern (i): no
$F_{\pm 1}$-containing interaction appears in $H$, so a
length-$3$ basis string of the form $F_{\pm1}\,[A_{(2)},D]\,D^\dagger$ has no second
generating placement to cancel it.

Next, consider $A_{(2)}=E_0$.
For $A_{(1)}=F_0$, the length-$3$ basis string
$F_0 F_{+2} F_{-2}$ is generated only by the placement with
$F_0 E_0\in X$ and $F_{+2}F_{-2}\in H$, so $q(F_0E_0)=0$.
For $A_{(1)}=C$ with $C\notin\{E_0,F_0,F_{\pm 1}\}$, the basis string $C\,[E_0,C]\,C^\dagger\propto CCC^\dagger$ is
generated only by the placement with $CE_0\in X$ and $CC^\dagger\in H$, so $q(CE_0)=0$.
Thus $A_{(2)}=E_0$ admits only $A_{(1)}=E_0$.

For $A_{(2)}=F_0$ and $A_{(1)}\ne F_0$, one considers the length-$3$ basis string
$A_{(1)} F_{+1}E_{-1}$ generated by the placement with $A_{(1)}F_0\in X$ and $E_{+1}E_{-1}\in H$; enumerating the possible cancellation partners as in Step 3 of pattern (i) and using the exclusions already established singles
out $A_{(1)}=F_0$ as the only admissible value. The remaining cases reduce
to doubling operators by the argument of
Step~4 of Sec.~\ref{ssec:k2-pi}. The coefficients of the doubling operators are then fixed to the form~\eqref{eq:k2-coef} exactly as in pattern (i).

\subsubsection{Patterns (iii), (v), and (vi)}
In these patterns one or more of the $f_m$ vanish while the rest are
nonzero. If $\vb A_i = A_{(1)}A_{(2)}$ contains an $F_{\pm m}$ corresponding
to a vanishing $f_{\pm m}$, the generating-placement argument of Step 1 of pattern (i) forces
the coefficient to vanish.  The analysis of the remaining basis strings proceeds as in the
proof of Ref.~\cite{Hokkyo2024}, and the conclusion is again that only the
doubling operators appearing in $\hloc$ can have nonzero coefficients, with the coefficients fixed to the form~\eqref{eq:k2-coef}.

\subsubsection{Pattern (iv): $f_{\pm 2}=0$, $f_0,f_{\pm 1}\ne 0$}
Write $\vb A_i=O_{m_1}O_{m_2}$ with $O_m\in\{E_m,F_m\}$ and $m_1,m_2\in\{0,\pm1,\pm2\}$, where the index $m$ denotes the weight of the operator [$O_{\pm2}=F_{\pm2}$]. Since $f_{\pm2}=0$, the argument of Step 1 of pattern (i) gives $q(\vb{A})_i = 0$ whenever $m_1 = \pm2$ or $m_2 = \pm2$; we may thus assume $|m_1|,|m_2|\le1$. We first show that $m_2=-m_1$ is necessary.
Suppose $q(O_{m_1}O_{m_2})\neq0$. Commuting with an interaction term $D_{+1}D_{-1}\in H$ [$D_{+1}\in\{E_{+1},F_{+1}\}$] placed on the right pair of sites produces a length-$3$ basis string
$O_{m_1}O'_{m_2+1}D_{-1}$ with $O'_{m_2+1}\propto[O_{m_2},D_{+1}]$ of weight $m_2+1$ (nonvanishing for a suitable choice of $D_{+1}$). Its only possible cancellation partners are of the form
\begin{equation}
\begin{array}{ccc}
& O_{m_3} & D_{-1}\\
O_{m_1} & O_{-m_1} & \\
\hline
O_{m_1} & O'_{m_2+1} & D_{-1}
\end{array},
\end{equation}
with $O_{m_1}O_{-m_1}\in H$, which requires $m_2+1=-m_1+m_3$. Since basis strings containing weight-$\pm2$ operators are already excluded, $|m_3|\le 1$, and hence $m_2\le -m_1$.
Repeating the argument with the lowering interaction $D_{-1}D_{+1}$ gives $m_2\ge -m_1$. Therefore
$m_2=-m_1$.

Next, for $\vb A_i = C_m D_{-m}$ with $D_{-m}\ne C_m^\dagger$, we show that
the coefficient still vanishes. Take $C_m D_{-m}=F_{+1}E_{-1}$ as a representative example.
Comparing the two generating placements for $\vb B=F_{+1}F_{-2}F_{+1}$
[note $[E_{-1},F_{-1}]=2F_{-2}$],
\begin{equation}
\begin{array}{ccc}
F_{+1} & E_{-1} & \\
       & F_{-1} & F_{+1} \\
\hline
F_{+1} & (2F_{-2}) & F_{+1}
\end{array},\;
\begin{array}{ccc}
& E_{-1} & F_{+1}\\
F_{+1} & F_{-1} & \\
\hline
F_{+1} & (2F_{-2}) & F_{+1}
\end{array},
\end{equation}
we obtain
\begin{equation}
f_{-1}\,q(F_{+1}E_{-1})_i = -f_{+1}\,q(E_{-1}F_{+1})_{i+1}.
\label{eq:contra1}
\end{equation}
Similarly, the placements generating $\vb B = F_{+1}F_{-1}F_0$,
$\vb B = E_{-1}E_{+1}F_0$, and $\vb B = E_{-1}E_0F_{+1}$ yield, respectively,
\begin{align}
3f_0\,q(F_{+1}E_{-1})_i &= f_{+1}\,q(E_0F_0)_{i+1},\label{eq:v}\\
3f_0\,q(E_{-1}F_{+1})_i &= e_{-1}\,q(E_0F_0)_{i+1},\label{eq:vi}\\
f_{-1}\,q(E_{-1}F_{+1})_i &= e_{-1}\,q(E_{-1}F_{+1})_{i+1}.\label{eq:vii}
\end{align}
Combining~\eqref{eq:vi} and \eqref{eq:vii} gives
$q(E_0F_0)_{i+1}=3f_0 f_{-1}^{-1}\, q(E_{-1}F_{+1})_{i+1}$,
which together with~\eqref{eq:v} (applied at site $i$) and $f_0\neq0$ gives
\begin{equation}
f_{-1}\,q(F_{+1}E_{-1})_i = +f_{+1}\,q(E_{-1}F_{+1})_{i+1}.
\label{eq:contra2}
\end{equation}
Since $f_{\pm1}\neq0$, Eqs.~\eqref{eq:contra1} and~\eqref{eq:contra2} are consistent only if
$q(F_{+1}E_{-1})_i=q(E_{-1}F_{+1})_{i+1}=0$. The coefficients of the remaining non-doubling weight-$0$ strings,
$E_{+1}F_{-1}$, $F_{-1}E_{+1}$, $E_0F_0$, and $F_0E_0$,
vanish by the same argument. Hence only doubling operators remain, and their coefficients are fixed to the form~\eqref{eq:k2-coef} as in pattern (i).

\section{Reduction to $3$-local quantities}
\label{sec:reduction}
To analyze the existence of $k$-local conserved quantities for $k\ge3$, we apply the theorems of Ref.~\cite{Hokkyo2025}, which reduce the problem to the analysis of $3$-local quantities.

To state the theorems, we introduce for each $k\geq2$ the space $\mathcal B^{(k)}$ of
strictly $k$-local quantities---linear combinations of basis strings of length
exactly $k$ (with arbitrary, possibly site-dependent coefficients)---together with its subspace
\begin{equation}
\mathcal B^{(k)}_\le:=\bigl\{\,\hat X\in\mathcal B^{(k)}\;\big|\;\len([\hat X,H])\le k\,\bigr\},
\label{eq:Ble}
\end{equation}
i.e., the strictly $k$-local quantities whose commutator with $H$ does not increase
the support length.
Note that the decomposition~\eqref{eq:H-decomp} is of the form required in Ref.~\cite{Hokkyo2025}: $\hloc_i$ is a sum of basis strings of length $2$ (traceless on both sites), and $H^{(1)}_i$ is traceless on site $i$.

\begin{theorem}[Hokkyo~\cite{Hokkyo2025}]
\label{thm:hokkyo}
Consider a Hamiltonian $H=\sum_i(\hloc_i+H^{(1)}_i)$ on a periodic chain of $N$ sites, with nearest-neighbor interactions $\hloc_i$ and
on-site potentials $H^{(1)}_i$, satisfying the following two assumptions:
\begin{itemize}[itemsep=2pt]
\item[\textup{(A)}] For every site $i$, any traceless operator $\hat X$ acting on site $i+1$, and any traceless operator $\hat Y$ acting on site $i$,
\begin{equation}
\begin{aligned}
&[I\otimes\hat X,\hloc_i]=0\;\Longrightarrow\;\hat X=0,\\
&[\hat Y\otimes I,\hloc_i]=0\;\Longrightarrow\;\hat Y=0.
\end{aligned}
\label{eq:inj-assump}
\end{equation}
\item[\textup{(B)}] $\mathcal B^{(2)}_\le=\mathbb C\,\sum_i\hloc_i$, i.e., every strictly $2$-local quantity $\hat X$ with $\len([\hat X,H])\le2$ is a scalar multiple of $\sum_i\hloc_i$.
\end{itemize}
If there is no
$3$-local quantity $Q$ such that $\len([Q,H])\le 2$,
then the system has no $k$-local conserved quantity for any
$3\le k\le N/2$. Equivalently, if there is a $k$-local conserved
quantity for some $3\le k\le N/2$, then there also exists a
$3$-local quantity $Q$ such that $\len([Q,H])\le 2$.
\end{theorem}

Note that the reduction is to $3$-local quantities $Q$ with $\len([Q,H])\le 2$, which form a strictly larger class than $3$-local \emph{conserved} quantities ($[Q,H]=0$); accordingly, in Secs.~\ref{sec:Ceq} and \ref{sec:3local} we rule out this larger class.

Let us verify the assumptions for the Hamiltonian~\eqref{eq:Hgen} with any of the patterns (i)--(vi).
Assumption (B) is precisely the content of Proposition~\ref{prop:doubling}: a strictly $2$-local $\hat X$ satisfies $\len([\hat X,H])\le2$ if and only if $[\hat X,H]$ contains no length-$3$ basis strings [$[\hat X,H]$ is at most $3$-local], and the proposition shows that this forces $\hat X\propto\sum_i\hloc_i$; conversely, $\len([\sum_i\hloc_i,H])=\len([h\sum_iF_{0,i},H])\le2$ by $\sum_i\hloc_i=H-h\sum_iF_{0,i}$.
For assumption (A), note that $[I\otimes\hat X,\hloc_i]=\sum_C c_C\, C\otimes[\hat X,C^\dagger]$, and the basis strings with distinct first factors $C$ are linearly independent; hence $[\hat X,C^\dagger]=0$ is required for every $C$ with $c_C\neq0$. Taking $C=E_0$ gives $[\hat X,E_0]=0$, so $\hat X$ has weight zero, $\hat X=xE_0+yF_0$. Taking $C=E_{+1}$ then gives $0=[\hat X,E_{-1}]=-xE_{-1}-3yF_{-1}$, so $x=y=0$. The second line of Eq.~\eqref{eq:inj-assump} follows in the same way. Note that only $e_0\neq0$ and $e_{+1}\neq0$ are used, so assumption (A) holds in all patterns (i)--(vi).

Hence, to rule out $k$-local conserved quantities for all $3\le k\le N/2$, it suffices to show that no $3$-local quantity $Q$ satisfies $\len([Q,H])\le2$.
The following theorem further restricts the form of such a $Q$.
\begin{theorem}[Hokkyo~\cite{Hokkyo2025}]
\label{thm:hokkyo2}
Consider a Hamiltonian $H$ as in Theorem~\ref{thm:hokkyo} satisfying assumption \textup{(A)}.
Then, for $3\le k\le N/2$, there exist linear isomorphisms
\begin{equation}
\iota_k\colon \mathcal B^{(k-1)}_\le\;\xrightarrow{\;\cong\;}\;\mathcal B^{(k)}_\le,
\label{eq:iota}
\end{equation}
which act on $\hat X^{(k-1)}=\sum_{i=1}^N\hat X^{(k-1)}_i\in\mathcal B^{(k-1)}_\le$ \textup(where $\hat X^{(k-1)}_i$ collects the basis strings with leftmost site $i$\textup) as
\begin{equation}
\hat X^{(k-1)}\;\longmapsto\;\sum_{i=1}^N\bigl[\hat X^{(k-1)}_i,\,\hloc_{i+k-2}\bigr]\in\mathcal B^{(k)}_\le.
\label{eq:iota-action}
\end{equation}
\end{theorem}

These theorems restrict the possible form of the strictly $3$-local part of a candidate $Q$ as follows.
Suppose $Q$ is a $3$-local quantity with $\len([Q,H])\le2$, and let $\hat X^{(3)}\in\mathcal B^{(3)}$ be its strictly $3$-local part.
Since the length-$4$ basis strings in $[Q,H]$ receive contributions only from $\hat X^{(3)}$, we have $\len([\hat X^{(3)},H])\le3$, i.e., $\hat X^{(3)}\in\mathcal B^{(3)}_\le$.
On the other hand, Proposition~\ref{prop:doubling} states precisely that $\mathcal B^{(2)}_\le$ is one dimensional and spanned by $\sum_i\hloc_i$ [an element $\hat X\in\mathcal B^{(2)}_\le$ satisfies the hypothesis of the proposition; conversely, using $\sum_i\hloc_i=H-h\sum_iF_{0,i}$, one finds $\len([\sum_i\hloc_i,H])=\len([h\sum_iF_{0,i},H])\le2$].
Applying Theorem~\ref{thm:hokkyo2} with $k=3$, we conclude that $\mathcal B^{(3)}_\le$ is spanned by
\begin{equation}
\begin{aligned}
\iota_3\Bigl(\sum_i\hloc_i\Bigr)&=\sum_i\bigl[\hloc_i,\hloc_{i+1}\bigr]\\
&=\sum_i\sum_{C_{(1)},C_{(2)}}\!\!c_{C_{(1)}}c_{C_{(2)}}\vb B_i,
\end{aligned}
\end{equation}
where $\vb B_i$ denotes the length-$3$ string
\begin{equation}
\vb B_i=
\begin{array}{@{}ccc@{}}
C_{(1)} & C_{(1)}^\dagger & \\
        & C_{(2)}         & C_{(2)}^\dagger
\end{array}
\;=\;C_{(1),i}\,[C_{(1)}^\dagger, C_{(2)}]_{i+1}\,\bigl(C_{(2)}^\dagger\bigr)_{i+2}.
\label{eq:3local-doubling}
\end{equation}
In other words, the strictly $3$-local part of $Q$ consists only of the \emph{doubling strings}~\eqref{eq:3local-doubling}, with coefficients
\begin{equation}
q\!\left(
\begin{array}{ccc}
C_{(1)} & C_{(1)}^\dagger & \\
        & C_{(2)}         & C_{(2)}^\dagger
\end{array}
\right)
=q_{k=3}\,c_{C_{(1)}}\,c_{C_{(2)}},
\label{eq:k3-coef}
\end{equation}
where $c_{C_{(j)}}$ is the coefficient of $C_{(j)}C_{(j)}^\dagger$ in
$\hloc$ and $q_{k=3}$ is a constant independent of the site $i$ and of $C_{(1)},C_{(2)}$.
If we can show that $q_{k=3}$ must vanish, then $Q$ has no strictly $3$-local part, contradicting $\len(Q)=3$; this is the goal of the next two sections.

\section{Necessary conditions for $3$-local quantities with $\len([Q,H])\le2$}
\label{sec:Ceq}
By Theorem~\ref{thm:hokkyo}, it remains to show that no $3$-local quantity $Q$ satisfies $\len([Q,H])\le2$.
Let $Q$ be such a quantity and write $Q=\hat X^{(3)}+\hat X^{(2)}+\hat X^{(1)}+q_0I$ with $\hat X^{(l)}\in\mathcal B^{(l)}$.
As shown in Sec.~\ref{sec:reduction}, $\hat X^{(3)}$ is fixed by Eqs.~\eqref{eq:3local-doubling} and \eqref{eq:k3-coef} up to the overall constant $q_{k=3}$.
The conditions that the length-$3$ basis strings in $[Q,H]$ vanish, $r_{\vb B^3}=0$, are linear in the coefficients of $\hat X^{(2)}$ and in $q_{k=3}$.
Following Ref.~\cite{Hokkyo2024}, one can eliminate the coefficients of $\hat X^{(2)}$ from these conditions and obtain ten equations of the form $q_{k=3}\times P_n(\{e_m,f_m,h\})=0$ $(n=1,\dots,10)$, where the $P_n$ are polynomials in the coupling constants.
Consequently, if a $3$-local quantity $Q$ with $\len([Q,H])\le2$ exists (so that $q_{k=3}\neq0$), the couplings must satisfy the ten conditions $P_n=0$, referred to as $(C\text{-}1)$--$(C\text{-}10)$ below.
These conditions were derived in Appendix~C of Ref.~\cite{Hokkyo2024} for the case in which all $e_m$ and $f_m$ are nonzero; the derivation, which we reproduce in Appendix~\ref{app:Ceq} for completeness, remains valid verbatim when some of the $f_m$ vanish, so the conditions apply to all patterns (i)--(vi) [see also the remark at the beginning of Appendix~\ref{app:Cfull} concerning misprints in four of the printed equations of Ref.~\cite{Hokkyo2024}].
The full list is collected in Appendix~\ref{app:Cfull}. In the analysis of
Sec.~\ref{sec:3local}, only the six equations
$(C\text{-}1),(C\text{-}2),(C\text{-}3),(C\text{-}4),(C\text{-}8),(C\text{-}10)$
are needed, and we display them here:

\begin{widetext}
\begin{align}
\tag{C-1}\label{eq:C1}
&\tfrac{3}{2}f_0|e_{+1}-f_{+1}|^2+e_0\Bigl(-\tfrac{1}{6}|e_{+1}|^2-\tfrac{1}{2}e_{+1}f_{-1}-\tfrac{1}{2}f_{+1}e_{-1}-\tfrac{1}{6}|f_{+1}|^2+\tfrac{1}{3}|f_{+2}|^2\Bigr)=0,\\
\tag{C-2}\label{eq:C2}
&(e_{+1}-f_{+1})\bigl[f_{+2}(e_0+3f_0+6h)-(e_{+1}^2+6e_{+1}f_{+1}+f_{+1}^2)\bigr]-f_{+2}^2(e_{-1}-f_{-1})=0,\\
\tag{C-3}\label{eq:C3}
&3f_0\!\bigl[f_{-2}(e_{+1}-f_{+1})^2+f_{+2}(e_{-1}-f_{-1})^2\bigr]+|f_{+2}|^2\!\Bigl(-\tfrac{4}{3}e_0^2+\tfrac{1}{3}|e_{+1}|^2+e_{+1}f_{-1}+f_{+1}e_{-1}+\tfrac{1}{3}|f_{+1}|^2-\tfrac{1}{3}|f_{+2}|^2\Bigr)=0,\\
\tag{C-4}\label{eq:C4}
&f_{+2}^2\bigl(|e_{+1}|^2+3e_{+1}f_{-1}-3f_{+1}e_{-1}-|f_{+1}|^2\bigr)+f_{+2}(e_{+1}\!+\!f_{+1})(e_{+1}\!-\!f_{+1})(e_0-21f_0-6h)+(e_{+1}\!+\!f_{+1})(e_{+1}\!-\!f_{+1})^3=0,\\
\tag{C-8}\label{eq:C8}
&(e_{+1}-f_{+1})\Bigl[\tfrac{3}{2}f_0\!\bigl(f_{-2}(e_{+1}-f_{+1})^2-f_{+2}(e_{-1}-f_{-1})^2\bigr)\nonumber\\
&\quad+|f_{+2}|^2\!\Bigl(\tfrac{1}{6}|f_{+2}|^2-\tfrac{1}{6}|e_{+1}|^2-\tfrac{1}{6}|f_{+1}|^2+\tfrac{3}{2}e_{+1}f_{-1}+\tfrac{3}{2}e_{-1}f_{+1}+\tfrac{e_0^2}{3}-18f_0^2-3e_0f_0+18f_0h\Bigr)\Bigr]\nonumber\\
&\quad+(e_{-1}-f_{-1})|f_{+2}|^2\bigl(3f_0f_{+2}-\tfrac{4}{3}e_{+1}f_{+1}\bigr)=0,\nonumber\\
\tag{C-10}\label{eq:C10}
&(e_{+1}+f_{+1})\Bigl[\tfrac{3}{2}f_0\!\bigl(3f_{-2}(e_{+1}-f_{+1})^2-f_{+2}(e_{-1}-f_{-1})^2\bigr)+|f_{+2}|^2\!\Bigl(\tfrac{1}{6}|f_{+2}|^2+\tfrac{e_0^2}{3}+6e_0f_0-27f_0^2\Bigr)\Bigr]\nonumber\\
&\quad+\tfrac{1}{6}(e_{+1}-f_{+1})|f_{+2}|^2\bigl(|f_{+1}|^2-|e_{+1}|^2+9e_{+1}f_{-1}-9e_{-1}f_{+1}\bigr)=0.\nonumber
\end{align}
\end{widetext}

Here, $|x|^2$ is understood as $x_{+m}x_{-m}$ for the Hermiticity-related pairs, e.g., $|f_{+2}|^2=f_{+2}f_{-2}$.
The four omitted equations,
$(C\text{-}5),(C\text{-}6),(C\text{-}7)$, and $(C\text{-}9)$, are not used below
but are listed in Appendix~\ref{app:Cfull} for completeness.

\section{Absence of $3$-local quantities with $\len([Q,H])\le2$}
\label{sec:3local}

We now show that the necessary conditions
$(C\text{-}1)$--$(C\text{-}10)$ cannot be satisfied in any of the six patterns (i)--(vi), as long as all $e_m$ are nonzero.  It is convenient to organize the six patterns into the following three cases according to which of $f_{\pm2}$, $f_{\pm1}$, and $f_0$ vanish:
\begin{itemize}[itemsep=2pt]
\item Case (a): $f_{\pm2}=0$, covering patterns (i), (iv), and (vi);
\item Case (b): $f_{\pm1}=0$ and $f_{\pm2}\neq0$, covering patterns (ii) and (v);
\item Case (c): $f_0=0$ and $f_{\pm1},f_{\pm2}\neq0$, covering pattern (iii).
\end{itemize}
Throughout this section, we use the shorthand
\begin{equation}
A := e_{+1}-f_{+1},\qquad B := e_{+1}+f_{+1},
\label{eq:AB}
\end{equation}
and
\begin{equation}
w := AB^\ast,\quad w_r := \mathrm{Re}\,w,\quad w_i := \mathrm{Im}\,w,
\label{eq:wdef}
\end{equation}
so that $|w|^2=|A|^2|B|^2$ [$w$ is used only in case (c)].  From $e_{+1}=(A+B)/2$ and
$f_{+1}=(B-A)/2$, together with $e_{-1}=e_{+1}^\ast$ and $f_{-1}=f_{+1}^\ast$, one verifies the identities
\begin{subequations}\label{eq:AB-ids}
\begin{align}
|e_{+1}|^2+|f_{+1}|^2 &= \tfrac{|A|^2+|B|^2}{2},\\
|e_{+1}|^2-|f_{+1}|^2 &= w_r,\\
e_{+1}f_{-1}+f_{+1}e_{-1} &= \tfrac{|B|^2-|A|^2}{2},\\
e_{+1}f_{-1}-f_{+1}e_{-1} &= i\,w_i,\\
e_{+1}f_{+1} &= \tfrac{B^2-A^2}{4},\\
e_{+1}^2+6e_{+1}f_{+1}+f_{+1}^2 &= 2B^2-A^2.
\end{align}
\end{subequations}

\subsection{Case (a): $f_{\pm2}=0$}
\label{ssec:caseA}

Setting $f_{\pm2}=0$ makes
$(C\text{-}3)$ and $(C\text{-}5)$--$(C\text{-}10)$ identically zero. The remaining
constraints $(C\text{-}1)$, $(C\text{-}2)$, and $(C\text{-}4)$ reduce, after multiplying $(C\text{-}1)$ by $6$ and $(C\text{-}2)$ by $-1$ and using the notation~\eqref{eq:AB}, to
\begin{align}
9f_0|A|^2 - e_0\bigl[|e_{+1}|^2+3(e_{+1}f_{-1}+f_{+1}e_{-1})+|f_{+1}|^2\bigr]&=0,\label{eq:a-C1}\\
A\,(e_{+1}^2+6e_{+1}f_{+1}+f_{+1}^2)&=0,\label{eq:a-C2}\\
B\,A^3&=0.\label{eq:a-C4}
\end{align}
If $A\ne0$, then~\eqref{eq:a-C4} gives $B=0$, i.e., $f_{+1}=-e_{+1}$.
In this case,
$e_{+1}^2+6e_{+1}f_{+1}+f_{+1}^2=-4e_{+1}^2$, so~\eqref{eq:a-C2} forces
$e_{+1}=0$, contradicting the assumption $e_{\pm1}\neq0$.  Therefore $A=0$, i.e.,
$e_{+1}=f_{+1}$; substituting $|f_{+1}|^2=|e_{+1}|^2$ and
$e_{+1}f_{-1}+f_{+1}e_{-1}=2|e_{+1}|^2$ into~\eqref{eq:a-C1} gives
$-8e_0|e_{+1}|^2=0$, so $e_0=0$ or $e_{+1}=0$, either of which contradicts the assumption $e_m\neq0$. Hence, the conditions $(C\text{-}1)$--$(C\text{-}10)$ cannot be satisfied in case (a).

\subsection{Case (b): $f_{\pm1}=0$ with $f_{\pm2}\neq0$}

Setting $f_{\pm1}=0$ gives $A=B=e_{+1}$ and
$e_{+1}f_{-1}+f_{+1}e_{-1}=0$. Multiplying $(C\text{-}1)$ by $6$, and dividing $(C\text{-}4)$ by $e_{+1}\neq0$, the constraints $(C\text{-}1)$, $(C\text{-}2)$, and $(C\text{-}4)$ become
\begin{align}
|e_{+1}|^2(9f_0-e_0)+2e_0|f_{+2}|^2 &=0,\label{eq:b-C1}\\
e_{+1}f_{+2}(e_0+3f_0+6h)-e_{+1}^3-f_{+2}^2 e_{-1} &=0,\label{eq:b-C2}\\
f_{+2}^2 e_{-1}+e_{+1}f_{+2}(e_0-21f_0-6h)+e_{+1}^3 &=0.\label{eq:b-C4}
\end{align}
Adding~\eqref{eq:b-C2} and~\eqref{eq:b-C4}, the terms $e_{+1}^3$ and
$f_{+2}^2 e_{-1}$ cancel, leaving $2e_{+1}f_{+2}(e_0-9f_0)=0$; since $e_{+1},f_{+2}\neq0$, we get $e_0=9f_0$.
Substituting this back into~\eqref{eq:b-C1} eliminates the $|e_{+1}|^2$ term,
yielding $2 e_0|f_{+2}|^2=0$ and hence $e_0=0$, contradicting the assumption $e_0\neq0$. Note that this argument applies to both pattern (ii) ($f_0\neq0$) and pattern (v) ($f_0=0$), since no assumption on $f_0$ was used. Hence, the conditions $(C\text{-}1)$--$(C\text{-}10)$ cannot be satisfied in case (b).

\subsection{Case (c): $f_0=0$ with $f_{\pm1},f_{\pm2}\ne0$}
\label{ssec:caseC}
Here we use $(C\text{-}1),(C\text{-}3),(C\text{-}8)$, and $(C\text{-}10)$ to derive a contradiction.

Substituting $f_0=0$ into $(C\text{-}1)$, dividing by $e_0\ne0$, and multiplying by $-6$, we obtain
\begin{equation}
2|f_{+2}|^2 = |e_{+1}|^2 + 3(e_{+1}f_{-1}+f_{+1}e_{-1}) + |f_{+1}|^2;
\label{eq:c-C1}
\end{equation}
similarly, substituting $f_0=0$ into $(C\text{-}3)$, dividing by $|f_{+2}|^2\ne 0$, and multiplying by $3$ gives
\begin{equation}
|e_{+1}|^2+3(e_{+1}f_{-1}+f_{+1}e_{-1})+|f_{+1}|^2-|f_{+2}|^2 = 4e_0^2.
\label{eq:c-C3}
\end{equation}
Substituting~\eqref{eq:c-C1} into~\eqref{eq:c-C3} gives $|f_{+2}|^2=4e_0^2$; combining this with~\eqref{eq:c-C1} and the identities~\eqref{eq:AB-ids} [which give $|e_{+1}|^2+3(e_{+1}f_{-1}+f_{+1}e_{-1})+|f_{+1}|^2=2|B|^2-|A|^2$], we obtain
\begin{equation}
|f_{+2}|^2 = 4 e_0^2,\qquad 2|B|^2-|A|^2 = 8 e_0^2.
\label{eq:case-c-basic}
\end{equation}

Setting $f_0=0$ in $(C\text{-}10)$, dividing by $|f_{+2}|^2/6$, and using~\eqref{eq:case-c-basic} and
the identities~\eqref{eq:AB-ids} yields
\begin{equation}
6e_0^2\,B + A(-w_r+9i\,w_i)=0.
\label{eq:c10-wform}
\end{equation}
If $A=0$, then~\eqref{eq:c10-wform} forces $B=0$ and hence
$e_{+1}=(A+B)/2=0$, contradicting our assumption.  Thus $A\ne0$.  Multiplying
\eqref{eq:c10-wform} by $A^\ast$ and using $A^\ast B=w^\ast=w_r-iw_i$,
the real and imaginary parts give (up to positive numerical factors)
\begin{equation}
w_r(6e_0^2-|A|^2)=0,\qquad w_i(3|A|^2-2e_0^2)=0.
\label{eq:c10-RI}
\end{equation}

Next, setting $f_0=0$ in $(C\text{-}8)$, dividing by $|f_{+2}|^2=4e_0^2\,(\neq0)$, and using the identities~\eqref{eq:AB-ids} to simplify
$-(|e_{+1}|^2+|f_{+1}|^2)+9(e_{+1}f_{-1}+f_{+1}e_{-1})=-5|A|^2+4|B|^2$ and
$8A^\ast e_{+1}f_{+1}=2A^\ast B^2-2|A|^2 A$, then using
$4|B|^2=16e_0^2+2|A|^2$ from~\eqref{eq:case-c-basic}, one obtains
\begin{equation}
A(22e_0^2-|A|^2) = 2A^\ast B^2.
\label{eq:C8-key}
\end{equation}
Multiplying by $A^\ast$ gives
\begin{equation}
|A|^2(22e_0^2-|A|^2) = 2(w^\ast)^2,
\label{eq:C8-w}
\end{equation}
whose left-hand side is real.  Hence $w^2\in\mathbb R$, equivalently
$w_rw_i=0$.

Combining $w_rw_i=0$ with~\eqref{eq:c10-RI}, there are three
sub-cases: (c-1) $w_r\ne0,\,w_i=0$; (c-2) $w_r=0,\,w_i\ne0$; and (c-3) $w_r=w_i=0$.

\emph{Sub-case (c-1): $w_r\ne0$, $w_i=0$.}
The first equation of~\eqref{eq:c10-RI} gives
$|A|^2=6e_0^2$. On the other hand, substituting $(w^\ast)^2=w_r^2=|A|^2|B|^2$ into
\eqref{eq:C8-w}, dividing by $|A|^2\neq0$, and using $2|B|^2=8e_0^2+|A|^2$ yields
$|A|^2=7e_0^2$. Comparing the two results gives $e_0=0$, a contradiction.

\emph{Sub-case (c-2): $w_r=0$, $w_i\ne0$.}
Now
$(w^\ast)^2=-w_i^2=-|A|^2|B|^2$, and~\eqref{eq:C8-w} after dividing by
$|A|^2$ becomes $22e_0^2-|A|^2=-2|B|^2=-(8e_0^2+|A|^2)$, i.e., $30 e_0^2=0$,
so $e_0=0$, a contradiction.

\emph{Sub-case (c-3): $w_r=w_i=0$.}
Then $w=AB^\ast=0$ and $A\ne 0$
imply $B=0$, so the second equation of~\eqref{eq:case-c-basic} reads $-|A|^2=8e_0^2$, which
is impossible because the left-hand side is negative while the
right-hand side is positive.

In all three sub-cases we reach a contradiction. Hence, the conditions $(C\text{-}1)$--$(C\text{-}10)$ cannot be satisfied in case (c) either.

In summary, in all of the patterns (i)--(vi), the necessary conditions $(C\text{-}1)$--$(C\text{-}10)$ are incompatible with $e_m\neq0$. Therefore $q_{k=3}=0$, i.e., there is no $3$-local quantity $Q$ with $\len([Q,H])\le2$, and by Theorem~\ref{thm:hokkyo} no $k$-local conserved quantity exists for any $3\le k\le N/2$. This proves the first statement of Theorem~\ref{thm:main}.
\section{$1$-local conserved quantities}
\label{sec:k1}

To complete the proof of Theorem~\ref{thm:main}, we finally determine the $1$-local conserved
quantities, following the analysis
of Ref.~\cite{Hokkyo2024}; for the frustration-free chain~\eqref{eq:Horig} the only one turns out to be $S^z$ (up to normalization and addition of the identity). A general $1$-local quantity reads
\begin{equation}
Q^{(1)}=\sum_i\sum_{m}\!\bigl(q_{E_m,i}E_{m,i}+q_{F_m,i}F_{m,i}\bigr),
\end{equation}
with site-dependent coefficients $q_{E_m,i}$ and $q_{F_m,i}$.
The condition $[Q^{(1)},H]=0$ requires the coefficient of every length-$2$ basis string in $[Q^{(1)},H]$ to vanish, which constrains the coefficients as follows.

The placements generating the length-$2$ basis string $E_{+1}E_{-1}$ in
$[Q^{(1)},H]$ are
\begin{equation}
\begin{array}{ccc}
& E_0 & \\
& E_{+1} & E_{-1}\\\hline
+ & E_{+1} & E_{-1}
\end{array},\quad
\begin{array}{ccc}
& & E_0 \\
& E_{+1} & E_{-1}\\\hline
- & E_{+1} & E_{-1}
\end{array}.
\end{equation}
Since $e_{+1}\ne 0$, the cancellation between the two placements forces
$q_{E_0,i}$ to be independent of $i$:
\begin{equation}
q_{E_0,i}=q_{E_0}.
\end{equation}
This is consistent with the conservation of $S^z=\sum_iE_{0,i}$.

Next, consider the coefficient $q_{E_{\pm1},i}$. The placements generating $F_0 F_{\pm 1}$ are
\begin{equation}
\begin{array}{ccc}
& E_{\pm 1} & \\
& F_{\mp 1} & F_{\pm 1}\\\hline
\pm 1 & F_0 & F_{\pm 1}
\end{array},\;
\begin{array}{ccc}
& & E_{\pm 1}\\
& F_0 & F_0 \\\hline
\mp 3 & F_0 & F_{\pm 1}
\end{array},
\end{equation}
and those generating $F_{\pm 1} F_0$ are
\begin{equation}
\begin{array}{ccc}
& E_{\pm 1} & \\
& F_0 & F_0\\\hline
\mp 3 & F_{\pm 1} & F_0
\end{array},\;
\begin{array}{ccc}
& & E_{\pm 1}\\
& F_{\pm 1} & F_{\mp 1}\\\hline
\pm 1 & F_{\pm 1} & F_0
\end{array}.
\end{equation}
The vanishing of these two coefficients gives the recursions
$f_{\mp1}q_{E_{\pm1},i}=3f_0\,q_{E_{\pm1},i+1}$ and
$3f_0\,q_{E_{\pm1},i}=f_{\pm1}\,q_{E_{\pm1},i+1}$, whose consistency requires
\begin{equation}
(9f_0^2-|f_{+1}|^2)\,q_{E_{\pm 1},i}=0.
\label{eq:k1-E1cond}
\end{equation}
Hence $q_{E_{\pm1},i}=0$ unless $|f_{+1}|^2=9f_0^2$.

By the same logic, for the coefficient $q_{F_0,i}$, the relevant generating placements for $F_{\pm 1}E_{\mp 1}$ are
\begin{equation}
\begin{array}{ccc}
& F_0 & \\
& E_{\pm 1} & E_{\mp 1}\\\hline
\pm 3 & F_{\pm 1} & E_{\mp 1}
\end{array},\;
\begin{array}{ccc}
& & F_0\\
& F_{\pm 1} & F_{\mp 1}\\\hline
\mp 3 & F_{\pm 1} & E_{\mp 1}
\end{array},
\end{equation}
and for $E_{\pm 1}F_{\mp 1}$ are
\begin{equation}
\begin{array}{ccc}
& F_0 & \\
& F_{\pm 1} & F_{\mp 1}\\\hline
\pm 3 & E_{\pm 1} & F_{\mp 1}
\end{array},\;
\begin{array}{ccc}
& & F_0\\
& E_{\pm 1} & E_{\mp 1}\\\hline
\mp 3 & E_{\pm 1} & F_{\mp 1}
\end{array},
\end{equation}
which yield the requirement
\begin{equation}
(e_{\pm 1}^2-f_{\pm 1}^2)\,q_{F_0,i}=0,
\end{equation}
i.e., $q_{F_0,i}=0$ unless $e_{\pm1}^2=f_{\pm1}^2$. Analogously, for the coefficients $q_{F_{\pm1},i}$, the placements generating $F_{\pm 1}E_0$
and $E_0 F_{\pm 1}$,
\begin{equation}
\begin{array}{ccc}
& F_{\pm 1} & \\
& E_0 & E_0\\\hline
\mp 1 & F_{\pm 1} & E_0
\end{array},\;
\begin{array}{ccc}
& & F_{\pm 1}\\
& F_{\pm 1} & F_{\mp 1}\\\hline
\pm 1 & F_{\pm 1} & E_0
\end{array},
\end{equation}
\begin{equation}
\begin{array}{ccc}
& F_{\pm 1} & \\
& F_{\mp 1} & F_{\pm 1}\\\hline
\pm 1 & E_0 & F_{\pm 1}
\end{array},\;
\begin{array}{ccc}
& & F_{\pm 1}\\
& E_0 & E_0\\\hline
\mp 1 & E_0 & F_{\pm 1}
\end{array},
\end{equation}
give the necessary condition
\begin{equation}
(e_0^2-|f_{+1}|^2)\,q_{F_{\pm1},i}=0.
\end{equation}
Finally, for the coefficients $q_{F_{\pm2},i}$, the generating placements (note the coefficients $\mp2$ arising from $[F_{\pm2},E_0]=\mp2F_{\pm2}$)
\begin{equation}
\begin{array}{ccc}
& F_{\pm 2} & \\
& E_0 & E_0\\\hline
\mp 2 & F_{\pm 2} & E_0
\end{array},\;
\begin{array}{ccc}
& & F_{\pm 2}\\
& F_{\pm 2} & F_{\mp 2}\\\hline
\pm 1 & F_{\pm 2} & E_0
\end{array},
\end{equation}
\begin{equation}
\begin{array}{ccc}
& F_{\pm 2} & \\
& F_{\mp 2} & F_{\pm 2}\\\hline
\pm 1 & E_0 & F_{\pm 2}
\end{array},\;
\begin{array}{ccc}
& & F_{\pm 2}\\
& E_0 & E_0\\\hline
\mp 2 & E_0 & F_{\pm 2}
\end{array},
\end{equation}
yield
\begin{equation}
(4e_0^2-|f_{+2}|^2)\,q_{F_{\pm2},i}=0.
\end{equation}

Collecting the above, the coefficient $q_{X,i}$ can be nonzero only if the corresponding condition below is satisfied:
\begin{align*}
&9f_0^2=|f_{+1}|^2\;\text{(for $E_{\pm 1}$)},\;
e_{\pm 1}^2=f_{\pm 1}^2\;\text{(for $F_0$)},\\
&e_0^2=|f_{+1}|^2\;\text{(for $F_{\pm 1}$)},\;
4e_0^2=|f_{+2}|^2\;\text{(for $F_{\pm 2}$)}.
\end{align*}
For the frustration-free chain~\eqref{eq:Horig}, we have
$f_{\pm 1}=f_{\pm 2}=0$, and the couplings~\eqref{eq:coup} satisfy $f_0\neq0$ ($\Delta\neq1$), $e_{\pm1}=-1\neq0$, and $e_0=-\Delta\neq0$.
The four conditions then reduce to $9f_0^2=0$, $e_{\pm1}^2=0$, $e_0^2=0$, and $4e_0^2=0$, respectively, none of which holds.
Hence $q_{E_{\pm1},i}=q_{F_0,i}=q_{F_{\pm1},i}=q_{F_{\pm2},i}=0$, and the only
surviving $1$-local conserved quantity is $q_{E_0}\sum_i E_{0,i}=q_{E_0}S^z$.
Combined with the results of Secs.~\ref{sec:k2}--\ref{sec:3local}, this completes the proof of Theorem~\ref{thm:main}: every conserved quantity of the chain~\eqref{eq:Horig} with $\len(Q)\le N/2$ is a linear combination of the identity, $S^z$, and $H$.

\section{Conclusion}
\label{sec:conclusion}
In this work, we analyzed the local conserved quantities of the $U(1)$-symmetric spin-$1$ chain~\eqref{eq:Hgen} for the six patterns (i)--(vi) of vanishing couplings $f_m$ listed in Sec.~\ref{sec:notation}, which are not covered by the previous study~\cite{Hokkyo2024}. We first showed that, in all six patterns, every $2$-local conserved quantity is a linear combination of the Hamiltonian, $1$-local conserved quantities, and the identity (Sec.~\ref{sec:k2}). This property allowed us to apply the reduction theorems of Ref.~\cite{Hokkyo2025}, so that the existence of $k$-local conserved quantities for all $3\le k\le N/2$ could be examined through $3$-local quantities whose commutator with $H$ is at most $2$-local (Sec.~\ref{sec:reduction}). We then showed that the necessary conditions $(C\text{-}1)$--$(C\text{-}10)$ for the existence of such a quantity cannot be satisfied in any of the six patterns (Secs.~\ref{sec:Ceq} and \ref{sec:3local}), and we determined the $1$-local conserved quantities (Sec.~\ref{sec:k1}). As a corollary, the frustration-free spin-$1$ chain~\eqref{eq:Horig}, which exhibits spontaneous $U(1)$ symmetry breaking at $T=0$ in one spatial dimension~\cite{arXiv:2310.16881}, admits no local conserved quantity other than linear combinations of the identity, $S^z$, and $H$, and in particular no local order parameter commuting with the Hamiltonian.

Physically, this resolves a question left open for the
model~\eqref{eq:Horig}.  Its $T=0$ ordering is an exception to
Coleman's theorem~\cite{Coleman1973}, yet---in contrast to the Heisenberg ferromagnet, the
well-known exception whose magnetization is itself a conserved
charge---it possesses no local order parameter that commutes with the
Hamiltonian.  Our proof makes this absence rigorous: the only nontrivial local
conserved quantity is $S^z$ itself, so the symmetry breaking cannot be
attributed to a conserved order parameter. This supports the view that the symmetry breaking of this model is enabled by the anomalously soft excitations that frustration-free systems generically host~\cite{arXiv:2406.06414,BravyiGosset2015,Gossetozgunov,Anshu,Lemm_2022,lemmCritical2025}.

It is instructive to examine, from this perspective, the fate of a commuting order parameter beyond the range $k\le N/2$.
Let $P$ denote the projector onto the $(2N+1)$-fold degenerate ground space of the frustration-free chain~\eqref{eq:Horig}.
Frustration-freeness implies $HP=PH=0$ [with the ground-state energy normalized to zero, as in Eq.~\eqref{eq:Horig}], so the Hermitian operator $PS^xP$, with $S^x:=\sum_iS^x_i$, commutes with $H$.
This operator is nonzero---its ground-space matrix elements are precisely the order parameter computed in Ref.~\cite{arXiv:2310.16881}---and it changes the eigenvalue of $S^z$ by $\pm1$, so it is not a linear combination (nor any polynomial) of the identity, $S^z$, and $H$.
Hence a symmetry-breaking order parameter commuting with the Hamiltonian does exist, but Theorem~\ref{thm:main} forces it to be nonlocal: $\len(PS^xP)>N/2$.
The $U(1)$ symmetry breaking of the chain~\eqref{eq:Horig} is thus accompanied only by such highly nonlocal commuting order parameters, in sharp contrast to the Heisenberg ferromagnet, whose commuting order parameter is a sum of single-site operators.

We also proved the absence of $k$-local conserved quantities with $3\le k\le N/2$ in the periodic Motzkin chain~\cite{Motzkin2504} (Appendix~\ref{app:motzkin}), which is also frustration-free, with an exactly $(2N+1)$-fold degenerate ground space, as we prove in Appendix~\ref{app:motzkin-gs}, thereby settling a conjecture of Ref.~\cite{Motzkin2504}. The Motzkin chain has likewise been reported to exhibit $U(1)$ symmetry breaking at $T=0$ in one spatial dimension~\cite{Menon2024}. Note, however, that this report concerns not the periodic version analyzed here but the original Motzkin chain, whose special boundary condition selects the equal-weight superposition of the Motzkin paths as the unique ground state~\cite{Bravyi2012Motzkin}. In that setting, the symmetry breaking is diagnosed through the long-range order $\langle S^x_iS^x_j\rangle\to4/9$ of transverse spin correlations in the unique ground state and the response to a symmetry-breaking field~\cite{Menon2024}, rather than through an exactly degenerate ground-state manifold as in the chain~\eqref{eq:Horig}. Thus, our result does not directly establish the nonexistence of a local order parameter commuting with the Hamiltonian of the original Motzkin chain. Nevertheless, it shows the broader relevance of the mechanism identified here: the periodic Motzkin chain provides a second frustration-free example, independent of the spin-$1$ chain~\eqref{eq:Horig}, that admits no nontrivial local conserved quantity in the range $3\le k\le N/2$---precisely the kind of quantity that could otherwise serve as an order parameter commuting with the Hamiltonian. Because the original and the periodic Motzkin chains differ essentially only in their boundary conditions, it is natural to expect that the $T=0$ symmetry breaking reported for the original chain is likewise not accompanied by a commuting local order parameter, but rather originates from the same frustration-free route around Coleman's theorem.

A rigorous extension of the proof to the original (non-periodic) Motzkin chain, the analysis of the remaining pattern $f_0=f_{\pm1}=f_{\pm2}=0$, and a general characterization of continuous symmetry breaking in frustration-free systems without a commuting order parameter remain important directions for future work.

\begin{acknowledgments}

We thank Rintaro Masaoka, Akihiro Hokkyo, Fuga Ishii, and Mizuki Sanatani for useful discussions.

\end{acknowledgments}

\appendix

\section{Derivation of $(C\text{-}1)$--$(C\text{-}10)$}
\label{app:Ceq}

In this appendix, we reproduce the derivation of the ten constraints listed in Appendix~\ref{app:Cfull}, following Appendix~C of Ref.~\cite{Hokkyo2024}, and confirm that it remains valid when some of the $f_m$ vanish.
As explained in Secs.~\ref{sec:reduction} and \ref{sec:Ceq}, we consider a $3$-local quantity $Q=\hat X^{(3)}+\hat X^{(2)}+\hat X^{(1)}+q_0I$ with $\len([Q,H])\le2$, whose strictly $3$-local part $\hat X^{(3)}$ is fixed by Eqs.~\eqref{eq:3local-doubling} and \eqref{eq:k3-coef}.
Without loss of generality, we may take $Q$ to be anti-Hermitian with $q_{k=3}$ real and normalized to $q_{k=3}=1$: indeed, since $\hat X^{(3)\dagger}=-\hat X^{(3)}$ for real $q_{k=3}$ [each $\sum_i[\hloc_i,\hloc_{i+1}]$ is anti-Hermitian], one of the combinations $Q-Q^\dagger$ and $i(Q+Q^\dagger)$, each of which again satisfies $\len([\,\cdot\,,H])\le2$, is anti-Hermitian and has a nonvanishing strictly $3$-local part whenever $Q$ does.

The constraints follow from the vanishing of the coefficients $r_{\vb B}$ of the length-$3$ basis strings $\vb B=O_{a}O_{b}O_{c}$ in $[Q,H]$, where $O_a\in\{E_a,F_a\}$ denotes a basis element of weight $a$.
Since $H$ has weight zero and $\hat X^{(3)}$, being built from $\hloc$, consists of weight-$0$ strings, the nontrivial conditions arise from the strings with $a+b+c=0$; the conditions from strings of nonzero total weight involve only the weight-$w\neq0$ part of $\hat X^{(2)}$ and force those coefficients to vanish~\cite{Hokkyo2024}.
The surviving coefficients of $\hat X^{(2)}$ that enter the weight-$0$ conditions are those of the non-doubling weight-$0$ strings. By analyzing the length-$3$ conditions arising from additional weight-$0$ strings in the same manner, these coefficients can be shown to be site independent and pairwise equal,
\begin{equation}
\begin{aligned}
&q_{E_{+1},F_{-1}}=q_{F_{+1},E_{-1}},\quad
q_{E_0,F_0}=q_{F_0,E_0},\\
&q_{E_{-1},F_{+1}}=q_{F_{-1},E_{+1}}.
\end{aligned}
\label{eq:q-pairings}
\end{equation}
We stress that the pairing~\eqref{eq:q-pairings} remains valid in all the patterns (i)--(vi), i.e., even when some of the $f_m$ vanish: it follows from the vanishing of the coefficients of only four weight-$0$ strings, $E_{+1}F_0E_{-1}$, $E_{-1}F_0E_{+1}$, $E_{+1}F_{-1}E_0$, and $E_0F_{+1}E_{-1}$, using only $e_{\pm1}\neq0$ and no condition on the $f_m$ [below we use the site independence and suppress the site indices; the same strings with leftmost site $i$ relate the coefficients at sites $i$ and $i+1$].
Indeed, the conditions from the first two strings yield
$e_{+1}\bigl(q_{E_{+1},F_{-1}}-q_{F_{+1},E_{-1}}\bigr)=0$ and
$e_{-1}\bigl(q_{E_{-1},F_{+1}}-q_{F_{-1},E_{+1}}\bigr)=0$, respectively, while for the latter two strings, which are mapped to each other by Hermitian conjugation combined with spatial inversion, the contributions of the strictly $3$-local part $\hat X^{(3)}$ cancel exactly in the sum of the two conditions, leaving
$3e_{+1}\bigl(q_{E_0,F_0}-q_{F_0,E_0}\bigr)+e_0\bigl(q_{E_{+1},F_{-1}}-q_{F_{+1},E_{-1}}\bigr)=0$.
Combining these three relations gives Eq.~\eqref{eq:q-pairings}.

We abbreviate these three independent coefficients as
\begin{equation}
q_1:=q_{E_{+1},F_{-1}},\quad q_2:=q_{E_0,F_0},\quad q_3:=q_{E_{-1},F_{+1}}.
\label{eq:q123-def}
\end{equation}
Anti-Hermiticity of $Q$ implies
\begin{equation}
q_3=-q_1^\ast,\qquad q_2\in i\mathbb R .
\label{eq:q3-herm}
\end{equation}

The weight-$0$ length-$3$ basis strings $\vb B$ yielding nontrivial conditions can be organized into eleven representative strings and their Hermitian conjugates; we denote the corresponding equations $r_{\vb B}=0$ by $E_1,\dots,E_{11}$ and $E_1^\dagger,\dots,E_{11}^\dagger$ [the string of $E_7$ is self-conjugate].
For $i=1,\ldots,11$ the equations can be written in the form
\begin{equation}
E_i:\quad c_{i1}(H)\,q_1+c_{i2}(H)\,q_2+c_{i3}(H)\,q_3-R_i(H)=0,
\label{eq:Ei-form}
\end{equation}
where $c_{ij}(H)$ are linear and $R_i(H)$ cubic polynomials in the coupling constants
$\{e_m,f_m,h\}$ [the $R_i$ originate from $\hat X^{(3)}$ and carry the overall factor $q_{k=3}=1$].  The equations $E_i^\dagger$ associated with
the Hermitian-conjugate strings are obtained as follows: denoting by $X^\dagger$ the polynomial obtained from $X$ by exchanging the indices $+m\leftrightarrow-m$ of the couplings, and using $q_1^\dagger=-q_3$,
$q_2^\dagger=-q_2$, and $q_3^\dagger=-q_1$, they read
\begin{equation}
E_i^\dagger:\quad -c_{i3}^\dagger q_1-c_{i2}^\dagger q_2-c_{i1}^\dagger q_3-R_i^\dagger=0 .
\label{eq:Ei-dagger}
\end{equation}

\begin{widetext}
The representative strings $\vb B$ and the coefficients $c_{ij}$ appearing in~\eqref{eq:Ei-form} are listed in the following table; the equation $E_7$
contains no $q_j$.
We note a sign convention: for $i=8$, the quantity $E_8$ as tabulated equals \emph{minus} the coefficient $r_{\vb B}$ of $\vb B=F_{+1}F_{+1}F_{-2}$ in $[Q,H]$ (both $c_{81}$ and $R_8$ carry a common overall sign relative to the direct commutator coefficient), while for all other rows $E_i=r_{\vb B}$. Since only the conditions $E_i=0$ enter the analysis, this convention does not affect any of the combinations below.
\begin{center}\small
\begin{tabular}{|c|l|c|c|c|}\hline
$i$ & $\vb B$ & $c_{i1}$ ($q_1$) & $c_{i2}$ ($q_2$) & $c_{i3}$ ($q_3$)\\\hline
$1$ & $E_{+1}E_{+1}F_{-2}$ & $f_{+2}$ & $0$ & $0$\\\hline
$2$ & $E_{+1}E_0 F_{-1}$ & $e_{+1}-f_{+1}$ & $0$ & $0$\\\hline
$3$ & $E_{+1}E_{-1}F_0$ & $3 f_0$ & $-e_{+1}$ & $0$\\\hline
$4$ & $E_{+1}F_{-1}E_0$ & $e_0$ & $-3 e_{+1}$ & $0$\\\hline
$5$ & $E_{+1}F_{-2}E_{+1}$ & $-2 e_{-1}$ & $0$ & $-2 e_{+1}$\\\hline
$6$ & $E_0 E_{+1}F_{-1}$ & $-e_0$ & $3 f_{+1}$ & $0$\\\hline
$7$ & $E_0 E_0 F_0$ & $0$ & $0$ & $0$\\\hline
$8$ & $F_{+1}F_{+1}F_{-2}$ & $f_{+2}$ & $0$ & $0$\\\hline
$9$ & $F_{+2}F_{-2}F_0$ & $0$ & $-2 f_{+2}$ & $0$\\\hline
$10$ & $F_{+1}F_{-1}F_0$ & $3 f_0$ & $-f_{+1}$ & $0$\\\hline
$11$ & $F_{+1}F_{-2}F_{+1}$ & $2 f_{-1}$ & $0$ & $2 f_{+1}$\\\hline
\end{tabular}
\end{center}

The right-hand sides $R_i$ read
\begin{align}
R_1 &= e_{+1}^3 + 2 e_{+1}^2 f_{+1} + \tfrac{e_{+1}e_0 f_{+2}}{2} - 3 e_{+1}f_{+1}^2 - \tfrac{9 e_{+1}f_{+2}f_0}{2} - 3 e_{+1}f_{+2}h\notag\\
&\quad + \tfrac{e_{-1}f_{+2}^2}{2} + \tfrac{3 e_0 f_{+1}f_{+2}}{2} - \tfrac{3 f_{+1}f_{+2}f_0}{2} + 3 f_{+1}f_{+2}h - \tfrac{f_{+2}^2 f_{-1}}{2},\notag\\
R_2 &= \tfrac{e_{+1}^2 e_0}{2} + \tfrac{15 e_{+1}^2 f_0}{2} + 3 e_{+1}^2 h - \tfrac{e_{+1}e_{-1}f_{+2}}{2} - \tfrac{3 e_{+1}f_{+2}f_{-1}}{2}\notag\\
&\quad + \tfrac{3 e_{-1}f_{+1}f_{+2}}{2} - \tfrac{e_0 f_{+1}^2}{2} - \tfrac{15 f_{+1}^2 f_0}{2} - 3 f_{+1}^2 h + \tfrac{f_{+1}f_{+2}f_{-1}}{2},\notag\\
R_3 &= \tfrac{e_{+1}^2 e_{-1}}{6} - \tfrac{3 e_{+1}^2 f_{-1}}{2} + \tfrac{e_{+1}e_{-1}f_{+1}}{2} - \tfrac{e_{+1}e_0^2}{3} + \tfrac{3 e_{+1}e_0 f_0}{2} + \tfrac{5 e_{+1}f_{+1}f_{-1}}{6} - \tfrac{e_{+1}f_{+2}f_{-2}}{6} + \tfrac{27 e_{+1}f_0^2}{2}\notag\\
&\quad - 9 e_{+1}f_0 h - \tfrac{3 e_{-1}f_{+2}f_0}{2} - \tfrac{3 e_0 f_{+1}f_0}{2} - \tfrac{9 f_{+1}f_0^2}{2} + 9 f_{+1}f_0 h + \tfrac{3 f_{+2}f_{-1}f_0}{2},\notag\\
R_4 &= -\tfrac{e_{+1}^2 e_{-1}}{2} + \tfrac{e_{+1}^2 f_{-1}}{2} - \tfrac{3 e_{+1}e_{-1}f_{+1}}{2} + \tfrac{e_{+1}e_0^2}{2} - \tfrac{9 e_{+1}e_0 f_0}{2} - 3 e_{+1}e_0 h + \tfrac{3 e_{+1}f_{+1}f_{-1}}{2}\notag\\
&\quad + \tfrac{e_{+1}f_{+2}f_{-2}}{2} + \tfrac{3 e_{-1}e_0 f_{+2}}{2} - \tfrac{e_0^2 f_{+1}}{2} - \tfrac{3 e_0 f_{+1}f_0}{2} + 3 e_0 f_{+1}h - \tfrac{3 e_0 f_{+2}f_{-1}}{2},\notag\\
R_5 &= -e_{+1}^2 f_{-2} - 3 e_{+1}e_0 f_{-1} + e_{+1}f_{+1}f_{-2} + 15 e_{+1}f_{-1}f_0 + 6 e_{+1}f_{-1}h\notag\\
&\quad + e_{-1}^2 f_{+2} + 3 e_{-1}e_0 f_{+1} - 15 e_{-1}f_{+1}f_0 - 6 e_{-1}f_{+1}h - e_{-1}f_{+2}f_{-1},\notag\\
R_6 &= -\tfrac{3 e_{+1}e_{-1}f_{+1}}{2} + \tfrac{e_{+1}e_0^2}{2} + \tfrac{3 e_{+1}e_0 f_0}{2} - 3 e_{+1}e_0 h + \tfrac{3 e_{+1}f_{+1}f_{-1}}{2} + \tfrac{3 e_{-1}e_0 f_{+2}}{2} - \tfrac{e_{-1}f_{+1}^2}{2}\notag\\
&\quad - \tfrac{e_0^2 f_{+1}}{2} + \tfrac{9 e_0 f_{+1}f_0}{2} + 3 e_0 f_{+1}h - \tfrac{3 e_0 f_{+2}f_{-1}}{2} + \tfrac{f_{+1}^2 f_{-1}}{2} - \tfrac{f_{+1}f_{+2}f_{-2}}{2},\notag\\
R_7 &= \tfrac{e_{+1}e_{-1}e_0}{3} - 3 e_{+1}e_{-1}f_0 + e_{+1}e_0 f_{-1} + 3 e_{+1}f_{-1}f_0 + e_{-1}e_0 f_{+1}\notag\\
&\quad + 3 e_{-1}f_{+1}f_0 + \tfrac{e_0 f_{+1}f_{-1}}{3} - \tfrac{2 e_0 f_{+2}f_{-2}}{3} - 3 f_{+1}f_{-1}f_0,\notag\\
R_8 &= -3 e_{+1}^2 f_{+1} + \tfrac{3 e_{+1}e_0 f_{+2}}{2} + 2 e_{+1}f_{+1}^2 - \tfrac{3 e_{+1}f_{+2}f_0}{2} + 3 e_{+1}f_{+2}h\notag\\
&\quad - \tfrac{e_{-1}f_{+2}^2}{2} + \tfrac{e_0 f_{+1}f_{+2}}{2} + f_{+1}^3 - \tfrac{9 f_{+1}f_{+2}f_0}{2} - 3 f_{+1}f_{+2}h + \tfrac{f_{+2}^2 f_{-1}}{2},\notag\\
R_9 &= 6 e_{+1}^2 f_0 + \tfrac{e_{+1}e_{-1}f_{+2}}{3} - 12 e_{+1}f_{+1}f_0 + e_{+1}f_{+2}f_{-1} + e_{-1}f_{+1}f_{+2}\notag\\
&\quad - \tfrac{4 e_0^2 f_{+2}}{3} + 6 f_{+1}^2 f_0 + \tfrac{f_{+1}f_{+2}f_{-1}}{3} - \tfrac{f_{+2}^2 f_{-2}}{3},\notag\\
R_{10} &= \tfrac{5 e_{+1}e_{-1}f_{+1}}{6} - \tfrac{3 e_{+1}e_0 f_0}{2} + \tfrac{e_{+1}f_{+1}f_{-1}}{2} - \tfrac{9 e_{+1}f_0^2}{2} + 9 e_{+1}f_0 h - \tfrac{3 e_{-1}f_{+1}^2}{2} + \tfrac{3 e_{-1}f_{+2}f_0}{2}\notag\\
&\quad - \tfrac{e_0^2 f_{+1}}{3} + \tfrac{3 e_0 f_{+1}f_0}{2} + \tfrac{f_{+1}^2 f_{-1}}{6} - \tfrac{f_{+1}f_{+2}f_{-2}}{6} + \tfrac{27 f_{+1}f_0^2}{2} - 9 f_{+1}f_0 h - \tfrac{3 f_{+2}f_{-1}f_0}{2},\notag\\
R_{11} &= -3 e_{+1}e_0 f_{-1} - e_{+1}f_{+1}f_{-2} + 15 e_{+1}f_{-1}f_0 + 6 e_{+1}f_{-1}h + 3 e_{-1}e_0 f_{+1}\notag\\
&\quad - 15 e_{-1}f_{+1}f_0 - 6 e_{-1}f_{+1}h + e_{-1}f_{+2}f_{-1} + f_{+1}^2 f_{-2} - f_{+2}f_{-1}^2 .\notag
\end{align}
\end{widetext}

The ten conditions $(C\text{-}1)$--$(C\text{-}10)$ of Appendix~\ref{app:Cfull} are obtained as the following linear combinations of the equations $E_i$ and $E_i^\dagger$, in which the coefficients $q_1,q_2,q_3$ cancel out, leaving pure constraints on the couplings:
\begin{equation*}
(C\text{-}1)=-\tfrac12 R_7=\tfrac12 E_7 .
\end{equation*}
\begin{equation*}
(C\text{-}2)=E_1-E_8 .
\end{equation*}
\begin{equation*}
(C\text{-}3)=-\tfrac12\bigl(f_{-2}E_9+f_{+2}E_9^\dagger\bigr) .
\end{equation*}
\begin{equation*}
(C\text{-}4)=2f_{+2}E_2-(e_{+1}-f_{+1})(E_1+E_8) .
\end{equation*}
\begin{equation*}
\begin{aligned}
(C\text{-}5)={}&f_{-2}(e_{-1}+f_{-1})(E_1+E_8)\\
&-f_{+2}(e_{+1}+f_{+1})(E_1^\dagger+E_8^\dagger)\\
&+f_{+2}f_{-2}(E_5-E_{11}) .
\end{aligned}
\end{equation*}
\begin{equation*}
\begin{aligned}
(C\text{-}6)={}&f_{-2}(f_{-1}-e_{-1})(E_1+E_8)\\
&-f_{+2}(f_{+1}-e_{+1})(E_1^\dagger+E_8^\dagger)\\
&-f_{+2}f_{-2}(E_5+E_{11}) .
\end{aligned}
\end{equation*}
\begin{equation*}
\begin{aligned}
(C\text{-}7)={}&f_{+2}f_{-2}(E_4+E_6)\\
&-\tfrac34(e_{+1}-f_{+1})\bigl(f_{-2}E_9-f_{+2}E_9^\dagger\bigr) .
\end{aligned}
\end{equation*}
\begin{equation*}
\begin{aligned}
(C\text{-}8)={}&f_{+2}f_{-2}(E_3-E_{10})\\
&-\tfrac14(e_{+1}-f_{+1})\bigl(f_{-2}E_9-f_{+2}E_9^\dagger\bigr) .
\end{aligned}
\end{equation*}
\begin{equation*}
\begin{aligned}
(C\text{-}9)={}&-e_0 f_{-2}(E_1+E_8)+f_{+2}f_{-2}(E_4-E_6)\\
&-\tfrac34(e_{+1}+f_{+1})\bigl(f_{-2}E_9-f_{+2}E_9^\dagger\bigr) .
\end{aligned}
\end{equation*}
\begin{equation*}
\begin{aligned}
(C\text{-}10)={}&-3f_0 f_{-2}(E_1+E_8)+f_{+2}f_{-2}(E_3+E_{10})\\
&-\tfrac14(e_{+1}+f_{+1})\bigl(f_{-2}E_9-f_{+2}E_9^\dagger\bigr) .
\end{aligned}
\end{equation*}

\section{Full list of $(C\text{-}1)$--$(C\text{-}10)$}
\label{app:Cfull}

For reference, we collect here the full set of ten necessary conditions for the existence of a $3$-local quantity $Q$ with $\len([Q,H])\le2$, following Ref.~\cite{Hokkyo2024}.  Equations $(C\text{-}1)$--$(C\text{-}4)$,
$(C\text{-}8)$, and $(C\text{-}10)$ are also displayed in Sec.~\ref{sec:Ceq}
(the only ones used in the main analysis); the four additional
equations, $(C\text{-}5)$, $(C\text{-}6)$, $(C\text{-}7)$, and $(C\text{-}9)$, are
listed here for completeness.
We have re-derived all ten conditions through the procedure of Appendix~\ref{app:Ceq}.  While $(C\text{-}1)$--$(C\text{-}4)$, $(C\text{-}7)$, and $(C\text{-}8)$ agree with the printed equations in Appendix~C of Ref.~\cite{Hokkyo2024}, our $(C\text{-}5)$, $(C\text{-}6)$, $(C\text{-}9)$, and $(C\text{-}10)$ differ from the printed versions there, which we believe contain misprints; in particular, we have checked that our $(C\text{-}9)$ and $(C\text{-}10)$ vanish on all the integrable families listed in Ref.~\cite{Hokkyo2024}, whereas the printed versions do not.

\begin{widetext}
\begin{align}
\tag{C-1}
&\tfrac{3}{2}f_0|e_{+1}-f_{+1}|^2+e_0\Bigl(-\tfrac{1}{6}|e_{+1}|^2-\tfrac{1}{2}e_{+1}f_{-1}-\tfrac{1}{2}f_{+1}e_{-1}-\tfrac{1}{6}|f_{+1}|^2+\tfrac{1}{3}|f_{+2}|^2\Bigr)=0,\\
\tag{C-2}
&(e_{+1}-f_{+1})\bigl[f_{+2}(e_0+3f_0+6h)-(e_{+1}^2+6e_{+1}f_{+1}+f_{+1}^2)\bigr]-f_{+2}^2(e_{-1}-f_{-1})=0,\\
\tag{C-3}
&3f_0\!\bigl[f_{-2}(e_{+1}-f_{+1})^2+f_{+2}(e_{-1}-f_{-1})^2\bigr]+|f_{+2}|^2\!\Bigl(-\tfrac{4}{3}e_0^2+\tfrac{1}{3}|e_{+1}|^2+e_{+1}f_{-1}+f_{+1}e_{-1}+\tfrac{1}{3}|f_{+1}|^2-\tfrac{1}{3}|f_{+2}|^2\Bigr)=0,\\
\tag{C-4}
&f_{+2}^2\bigl(|e_{+1}|^2+3e_{+1}f_{-1}-3f_{+1}e_{-1}-|f_{+1}|^2\bigr)+f_{+2}(e_{+1}\!+\!f_{+1})(e_{+1}\!-\!f_{+1})(e_0-21f_0-6h)+(e_{+1}\!+\!f_{+1})(e_{+1}\!-\!f_{+1})^3=0,\\
\tag{C-5}
&\bigl(|e_{+1}+f_{+1}|^2-|f_{+2}|^2\bigr)\bigl(f_{+2}(e_{-1}-f_{-1})^2-f_{-2}(e_{+1}-f_{+1})^2\bigr)=0,\\
\tag{C-6}
&(10\,e_0-42\,f_0-12\,h)(f_{+1}e_{-1}-e_{+1}f_{-1})\,|f_{+2}|^2\nonumber\\
&\quad+\bigl(f_{+2}(e_{-1}^2-f_{-1}^2)-f_{-2}(e_{+1}^2-f_{+1}^2)\bigr)\bigl(|f_{+2}|^2-|e_{+1}-f_{+1}|^2\bigr)=0,\nonumber\\
\tag{C-7}
&(e_{+1}-f_{+1})\Bigl[\tfrac{9}{2}f_0\bigl(f_{-2}(e_{+1}-f_{+1})^2-f_{+2}(e_{-1}-f_{-1})^2\bigr)-|f_{+2}|^2\bigl(\tfrac{1}{2}|f_{+2}|^2+e_0(e_0-3f_0-6h)\bigr)\Bigr]\nonumber\\
&\quad+(e_{-1}-f_{-1})|f_{+2}|^2\bigl(-3e_0f_{+2}+\tfrac{1}{2}(e_{+1}^2+6e_{+1}f_{+1}+f_{+1}^2)\bigr)=0,\nonumber\\
\tag{C-8}
&(e_{+1}-f_{+1})\Bigl[\tfrac{3}{2}f_0\bigl(f_{-2}(e_{+1}-f_{+1})^2-f_{+2}(e_{-1}-f_{-1})^2\bigr)\nonumber\\
&\quad+|f_{+2}|^2\Bigl(\tfrac{1}{6}|f_{+2}|^2-\tfrac{1}{6}|e_{+1}|^2-\tfrac{1}{6}|f_{+1}|^2+\tfrac{3}{2}e_{+1}f_{-1}+\tfrac{3}{2}e_{-1}f_{+1}+\tfrac{e_0^2}{3}-18f_0^2-3e_0f_0+18f_0h\Bigr)\Bigr]\nonumber\\
&\quad+(e_{-1}-f_{-1})|f_{+2}|^2\bigl(3f_0f_{+2}-\tfrac{4}{3}e_{+1}f_{+1}\bigr)=0,\nonumber\\
\tag{C-9}
&(e_{+1}+f_{+1})\Bigl[(e_0+\tfrac{9}{2}f_0)f_{-2}(e_{+1}-f_{+1})^2-\tfrac{9}{2}f_0f_{+2}(e_{-1}-f_{-1})^2+|f_{+2}|^2\bigl(-\tfrac{|f_{+2}|^2}{2}+2e_0^2+\tfrac{|e_{+1}-f_{+1}|^2}{2}\bigr)\Bigr]=0,\\
\tag{C-10}
&(e_{+1}+f_{+1})\Bigl[\tfrac{3}{2}f_0\bigl(3f_{-2}(e_{+1}-f_{+1})^2-f_{+2}(e_{-1}-f_{-1})^2\bigr)+|f_{+2}|^2\bigl(\tfrac{1}{6}|f_{+2}|^2+\tfrac{e_0^2}{3}+6e_0f_0-27f_0^2\bigr)\Bigr]\nonumber\\
&\quad+\tfrac{1}{6}(e_{+1}-f_{+1})|f_{+2}|^2\bigl(|f_{+1}|^2-|e_{+1}|^2+9e_{+1}f_{-1}-9e_{-1}f_{+1}\bigr)=0.\nonumber
\end{align}
\end{widetext}

These equations remain valid for any choice of vanishing $f_m$, since no division by $f_m$ is involved in their derivation.

\section{Nonintegrability of the area-weighted periodic Motzkin chain}
\label{app:motzkin}

\subsection{Definition of the model}
\label{ssec:motzkin-model}
The Motzkin chain is a frustration-free spin chain whose unique ground state is the equal-weight superposition of Motzkin paths~\cite{Bravyi2012Motzkin}. This system also has a $U(1)$ symmetry and has been reported to exhibit spontaneous symmetry breaking at $T=0$~\cite{Menon2024}. While a system related to the Motzkin chain was proved to be integrable~\cite{HaoSalbergerKorepin2023}, we show in this appendix the nonintegrability of the Motzkin chain under the periodic boundary condition~\cite{Motzkin2504}, in the generalized form with area weight $t$~\cite{ZAK17,LevineMovassagh2017,Andrei2022}:
\begin{equation}
H^{(t)}=\sum_{i=1}^N H^{(t)}_i,\qquad
H^{(t)}_i = \Pi^{(t)}_{U,i}+\Pi^{(t)}_{D,i}+\Pi^{(t)}_{F,i},
\label{eq:Ht-def}
\end{equation}
where the projectors act on sites $i,i+1$ and read (with the normalization factor
$1/(1+t^2)$ omitted)
\begin{align}
\Pi^{(t)}_{U,i}&=|\Phi_t\rangle\langle\Phi_t|_{i,i+1}, & |\Phi_t\rangle&=|u\,0\rangle-t|0\,u\rangle,\notag\\
\Pi^{(t)}_{D,i}&=|\Psi_t\rangle\langle\Psi_t|_{i,i+1}, & |\Psi_t\rangle&=|0\,d\rangle-t|d\,0\rangle,\label{eq:UDF-t}\\
\Pi^{(t)}_{F,i}&=|\Theta_t\rangle\langle\Theta_t|_{i,i+1}, & |\Theta_t\rangle&=|u\,d\rangle-t|0\,0\rangle,\notag
\end{align}
with $t\in\mathbb R\setminus\{0\}$. Here $|u\rangle:=|{+}1\rangle$, $|0\rangle$, and $|d\rangle:=|{-}1\rangle$ denote the eigenstates of $E_0$ with eigenvalues $+1,0,-1$, respectively. At $t=1$ the model reduces to the standard periodic
Motzkin chain~\cite{Motzkin2504}. In what follows we use the operator basis
$\mathcal O=\{I,E_0,E_{\pm1},F_0,F_{\pm1},F_{\pm2}\}$ of
Eq.~\eqref{eq:basis}.

Expanding~\eqref{eq:UDF-t} in the $\mathcal O$ basis and rearranging
the terms in the translation sum (dropping an overall constant, i.e., a multiple of the identity, which does not affect conserved quantities), we write
\begin{equation}
H^{(t)}=\sum_{i=1}^N \bigl(\hloc_{i}+H^{(1)}_i\bigr)+\text{const}.
\label{eq:H-split}
\end{equation}
Here the nearest-neighbor interaction term $\hloc_i$, acting on sites $i,i+1$, is
\begin{equation}
\begin{aligned}
\hloc=&\;e_0\,E_0\otimes E_0+f_0\,F_0\otimes F_0\\
&+\sigma\,E_0\otimes F_0+\tau\,F_0\otimes E_0\\
&+e_1\,(E_{+1}\otimes E_{-1}+E_{-1}\otimes E_{+1})\\
&+f_1\,(F_{+1}\otimes F_{-1}+F_{-1}\otimes F_{+1})\\
&+\mu_1\,(E_{+1}\otimes F_{-1}+E_{-1}\otimes F_{+1})\\
&+\mu_2\,(F_{+1}\otimes E_{-1}+F_{-1}\otimes E_{+1}),
\end{aligned}
\label{eq:H-param}
\end{equation}
where the first factor of each tensor product acts on site $i$ and the second on site $i+1$, and the on-site term $H^{(1)}_i$ is
\begin{equation}
H^{(1)}_i=h_{F_0}^{(t)}\,F_{0,i},\qquad h_{F_0}^{(t)}=-\tfrac{t^2}{3}\ ,
\label{eq:H-onebody}
\end{equation}
the identity shift being $(t^2+1)/3$ per site.
Note that, unlike the Hamiltonian~\eqref{eq:Hgen} analyzed in the main text, $\hloc$ contains the cross terms with coefficients $\sigma,\tau,\mu_1,\mu_2$; the analysis of this appendix is therefore independent of, though parallel to, that of the main text.  Throughout this appendix, the symbols $\hloc_i$ and $e_0,f_0,e_1,f_1$ refer to the Motzkin quantities defined here, not to those of Eq.~\eqref{eq:Hgen}.
The values of the eight two-body coefficients
$(e_0,f_0,\sigma,\tau,e_1,f_1,\mu_1,\mu_2)$ are
\begin{equation}
\begin{aligned}
e_0&:=-\tfrac{1}{4}, & f_0&:=-\tfrac{1}{12},\\
\sigma&:=\tfrac{2t^2-1}{12}, & \tau&:=\tfrac{1-2t^2}{12},\\
e_1&:=-\tfrac{3t}{4}, & f_1&:=-\tfrac{t}{4},\\
\mu_1&:=+\tfrac{t}{4}, & \mu_2&:=-\tfrac{t}{4}.
\end{aligned}
\label{eq:t-params}
\end{equation}
At $t=1$ these become
$(e_0,f_0,\sigma,\tau,e_1,f_1,\mu_1,\mu_2)=(-\tfrac14,-\tfrac1{12},\tfrac1{12},-\tfrac1{12},-\tfrac34,-\tfrac14,\tfrac14,-\tfrac14)$,
matching the standard periodic Motzkin chain~\cite{Motzkin2504}.

\subsection{Weight decomposition and $U(1)$ symmetry}
\label{ssec:motzkin-weight}

Recall the weight $w(O)$ of a basis element $O\in\mathcal O$ introduced in Eq.~\eqref{eq:weight-def}.
For an $N$-site tensor product of basis elements we define the total
weight
\begin{equation}
W(O_{a_1}\otimes\cdots\otimes O_{a_N}):=\sum_{i=1}^{N} w(O_{a_i}).
\label{eq:weight-tensor}
\end{equation}
Writing $S^z:=\sum_i E_{0,i}$, the operator algebra decomposes
into weight sectors $\mathcal V_W:=\{X:[S^z,X]=W X\}$, and the
weight is additive under products and commutators.

Inspecting the Hamiltonian~\eqref{eq:H-param}, every basis string with a nonzero coefficient satisfies
\[
w(O_a)+w(O_b)=0
\]
($E_0E_0,F_0F_0,E_0F_0,F_0E_0$ give $0+0$;
$E_{\pm1}E_{\mp1},F_{\pm1}F_{\mp1},E_{\pm1}F_{\mp1},F_{\pm1}E_{\mp1}$
give $\pm1\mp1=0$), so $\hloc\in\mathcal V_0$. Likewise
$H^{(1)}_i=h_{F_0}^{(t)}F_{0,i}\in\mathcal V_0$. Therefore, we have the following.

\begin{lemma}[The Hamiltonian has weight $0$]\label{lem:H-weight-zero}
$H^{(t)}=\sum_i H^{(t)}_i$ belongs to $\mathcal V_0$. In particular,
\begin{equation}
[S^z,H^{(t)}]=0,
\label{eq:H-Sz-commute}
\end{equation}
i.e., $S^z$ is a conserved quantity of the area-weighted periodic
Motzkin chain.
\end{lemma}

\subsection{Check of the prerequisites for Theorem~\ref{thm:hokkyo}}
\label{ssec:motzkin-prereq}

To apply Theorems~\ref{thm:hokkyo} and \ref{thm:hokkyo2}, the system must satisfy the
following two conditions:
\begin{itemize}\setlength{\itemsep}{2pt}
  \item[(A)] For any traceless one-site operators $X,Y$,
  \begin{align}
    &\big[I\otimes X,\,\hloc_i\big]=0\Rightarrow X=0,\notag\\
    &\big[Y\otimes I,\,\hloc_i\big]=0\Rightarrow Y=0.
    \label{eq:inj}
  \end{align}
  \item[(B)] $\mathcal B^{(2)}_\le=\mathbb C\sum_i\hloc_i$ for $H^{(t)}$, i.e., every strictly $2$-local quantity $X$ with $\len([X,H^{(t)}])\le2$ is a scalar multiple of $\sum_i\hloc_i$.
\end{itemize}

\subsubsection{Proof of (A)}
\label{sssec:motzkin-inj}

We first show that (A) holds under the conditions
$f_1\neq0$ and $e_1 f_1-\mu_1\mu_2\neq 0$, both of which are satisfied by the couplings~\eqref{eq:t-params} for all $t\neq0$ [$f_1=-t/4$ and $e_1f_1-\mu_1\mu_2=t^2/4$]. We spell out $\big[I\otimes X,\,\hloc_i\big]=0\Rightarrow X=0$; the statement $\big[Y\otimes I,\,\hloc_i\big]=0\Rightarrow Y=0$ can be proved in the same way.

\begin{proof}
Decompose $X$ by weight, $X=\sum_W X_W$:
\begin{equation}
\begin{aligned}
X_0&=x_{E_0}E_0+x_{F_0}F_0,\\
X_{+1}&=x_{E_{+1}}E_{+1}+x_{F_{+1}}F_{+1},\\
X_{-1}&=x_{E_{-1}}E_{-1}+x_{F_{-1}}F_{-1},\\
X_{+2}&=x_{F_{+2}}F_{+2},\quad X_{-2}=x_{F_{-2}}F_{-2}.
\end{aligned}
\end{equation}
Since $\hloc\in\mathcal V_0$, we have $[I\otimes X_W,\hloc]\in\mathcal V_W$,
so the vanishing condition is imposed independently in each weight
sector. We treat the sectors in turn.

\subsubsection*{$W=+2$ sector.}
Using $[F_{+2},F_{-1}]=-E_{+1}$, the coefficient of the basis string $F_{+1}E_{+1}$ in $[I\otimes x_{F_{+2}}F_{+2}, \hloc_i]$ is $-f_1x_{F_{+2}}$, and no other term of $\hloc$ contributes to this string. Hence $f_1\neq0$ forces
$x_{F_{+2}}=0$, and the claim holds in the $W=+2$ sector.

\subsubsection*{$W=+1$ sector.}
For $X_{+1}=x_{E_{+1}}E_{+1}+x_{F_{+1}}F_{+1}$, the column expressions
\begin{equation}
    \begin{array}{cc}
         & E_{+1} \\
          F_{-1} & F_{+1} \\
          \hline
          F_{-1} & (-2F_{+2})
    \end{array}, \
    \begin{array}{cc}
         & F_{+1} \\
          F_{-1} & E_{+1} \\
          \hline
          F_{-1} & (2F_{+2})
    \end{array},
\end{equation}
\begin{equation}
    \begin{array}{cc}
         & E_{+1} \\
          E_{-1} & F_{+1} \\
          \hline
          E_{-1} & (-2F_{+2})
    \end{array}, \
    \begin{array}{cc}
         & F_{+1} \\
          E_{-1} & E_{+1} \\
          \hline
          E_{-1} & (2F_{+2})
    \end{array},
\end{equation}
give
\begin{align}
    f_1\,x_{E_{+1}} - \mu_2\,x_{F_{+1}} &= 0,\\
    \mu_1\,x_{E_{+1}} - e_1\,x_{F_{+1}} &= 0 .
\end{align}
This has only the trivial solution when $e_1 f_1-\mu_1\mu_2\neq 0$, so
the claim holds in the $W=+1$ sector.

\subsubsection*{$W=0$ sector.}
By the same reasoning as the $W=+1$ sector, the column expressions
\begin{equation}
    \begin{array}{cc}
         & E_{0} \\
          E_{-1} & E_{+1} \\
          \hline
          E_{-1} & E_{+1}
    \end{array}, \
    \begin{array}{cc}
         & F_{0} \\
          E_{-1} & F_{+1} \\
          \hline
          E_{-1} & (3E_{+1})
    \end{array},
\end{equation}
\begin{equation}
    \begin{array}{cc}
         & E_{0} \\
          F_{-1} & E_{+1} \\
          \hline
          F_{-1} & E_{+1}
    \end{array}, \
    \begin{array}{cc}
         & F_{0} \\
          F_{-1} & F_{+1} \\
          \hline
          F_{-1} & (3E_{+1})
    \end{array},
\end{equation}
show that the claim holds when
\begin{equation}
    e_1 f_1-\mu_1\mu_2\neq0 .
\end{equation}

\subsubsection*{$W<0$ sectors.}
Since $(\hloc_i)^\dagger=\hloc_i$, validity of (A) in the $W<0$
sectors is equivalent to validity in the $W>0$ sectors. As we have
already verified (A) in the $W>0$ sectors, it holds automatically
here.

Therefore (A) holds in all sectors under the conditions
\begin{equation}
    f_1\neq0\quad\text{and}\quad e_1 f_1-\mu_1\mu_2\neq 0 .
\end{equation}
\end{proof}

\subsubsection{Proof of (B)}
\label{sssec:motzkin-lemB}

We consider the case in which the conditions
\begin{equation}
\text{(M1)}\ e_0 f_0\ne\sigma\tau,\quad
\text{(M2)}\ \mu_1\ne\mu_2,\quad
\text{(M3)}\ e_1 f_1\ne\mu_1\mu_2
\label{eq:nondegen-cond}
\end{equation}
hold [labeled (M1)--(M3) to avoid confusion with the patterns (i)--(vi) of the main text]. The couplings~\eqref{eq:t-params} meet these conditions for all $t\neq0$: (M1) $e_0f_0-\sigma\tau=[3+(2t^2-1)^2]/144>0$, (M2) $\mu_1-\mu_2=t/2\neq0$, and (M3) $e_1f_1-\mu_1\mu_2=t^2/4\neq0$.

As in Sec.~\ref{sec:k2}, it suffices to show that a strictly $2$-local quantity $X=\sum_i\sum_{O_mO_n}q_{O_mO_n,i}(O_m)_i(O_n)_{i+1}$ whose commutator $[X,H^{(t)}]$ contains no length-$3$ basis strings must be proportional to $\sum_i\hloc_i$; here $m,n$ denote the weights of the basis elements $O_m,O_n\in\mathcal O$.
Since $H^{(t)}$ has weight zero (Lemma~\ref{lem:H-weight-zero}), each weight component of $X$ can be analyzed separately, and we treat the components according to their total weight $W=m+n$.

First, we show
that whenever one of $O_m,O_n$ has weight $\pm2$, the coefficient
$q_{O_mO_n,i}$ vanishes.
Assume $O_m=F_{\pm2}$ (the case $O_n=F_{\pm2}$ is analogous). Since the Hamiltonian contains no term with
$F_{\pm2}$, a length-$3$ basis string of the form $F_{\pm2}\,[O_n,O_a]\,O_{-a}$ is
generated only by placements of the form
\begin{equation}
    \begin{array}{ccc}
        F_{\pm2} & O_n & \\
                 & O_a & O_{-a} \\
        \hline
    \end{array}
    , \ \
    \begin{array}{ccc}
        F_{\pm2} & O_n' & \\
                 & O_a' & O_{-a} \\
        \hline
    \end{array},
\end{equation}
i.e., only by placements in which the Hamiltonian term sits on the right pair of sites. The cancellation condition among these placements has exactly the form of condition
(A) applied to the operator $\sum_{O_n}q_{F_{\pm2}O_n,i}O_n$; hence, when (A) holds, the coefficients
$q_{F_{\pm2}O_n,i}$ necessarily vanish.

Next, we show that the coefficients vanish in the $W>0$ sectors. Since
the strings containing $F_{\pm2}$ have already been excluded, it remains to treat the cases where $O_m,O_n\in\{E_{\pm1},F_{\pm1},E_0,F_0\}$ with $m+n>0$, i.e., where at least one of
$O_m,O_n$ is $E_{+1}$ or $F_{+1}$. When $O_n\in\{E_{+1},F_{+1}\}$, the
length-$3$ strings of the form $O_mF_{+2}O_{-1}$ are generated only by placements
in which the Hamiltonian term sits on the right pair of sites; the resulting conditions have the same form as the (A) condition in
the $W=+1$ sector [cf.\ the column expressions of Sec.~\ref{sssec:motzkin-inj}]. The case
$O_m\in\{E_{+1},F_{+1}\}$ is analogous, using placements with the Hamiltonian on the left pair. Hence, as long as condition (M3) holds, all $W>0$ coefficients vanish.
The $W<0$ coefficients vanish by the same
argument combined with $(\hloc_i)^\dagger=\hloc_i$.

Finally, we examine the $W=0$ sector. Because the linear conditions on the coefficients are translation covariant, the solution space decomposes into momentum sectors, and we may take
\begin{equation}
q_{ab,i}=\omega^{ki} q_{ab}^{(k)},\qquad
\omega:=e^{2\pi i/N},\quad k\in\{0,1,\ldots,N-1\}.
\label{eq:fourier-def}
\end{equation}
Having shown that the coefficients
of $F_{+2}F_{-2}$ and $F_{-2}F_{+2}$ vanish, we now completely determine the remaining
twelve $W=0$ components, split into three blocks:
\begin{equation}
\begin{aligned}
\bm a&:=(a_1,a_2,a_3,a_4)\\
&\phantom{:}=(q_{E_0E_0}^{(k)},q_{E_0F_0}^{(k)},q_{F_0E_0}^{(k)},q_{F_0F_0}^{(k)}),\\
\bm b&:=(b_1,b_2,b_3,b_4)\\
&\phantom{:}=(q_{E_{+1}E_{-1}}^{(k)},q_{E_{+1}F_{-1}}^{(k)},q_{F_{+1}E_{-1}}^{(k)},q_{F_{+1}F_{-1}}^{(k)}),\\
\bm c&:=(c_1,c_2,c_3,c_4)\\
&\phantom{:}=(q_{E_{-1}E_{+1}}^{(k)},q_{E_{-1}F_{+1}}^{(k)},q_{F_{-1}E_{+1}}^{(k)},q_{F_{-1}F_{+1}}^{(k)}).
\end{aligned}
\label{eq:vars-blocks}
\end{equation}

The vanishing of the coefficients of the four length-$3$ basis strings
$E_{+1}E_{-1}E_0$, $E_{+1}E_{-1}F_0$, $F_{+1}E_{-1}E_0$, and
$F_{+1}E_{-1}F_0$ in $[X,H^{(t)}]$ gives
\begin{equation}
\begin{aligned}
e_0 b_1+3\tau b_2&=\omega^k(e_1 a_1+3\mu_1 a_3),\\
\sigma b_1+3 f_0 b_2&=\omega^k(e_1 a_2+3\mu_1 a_4),\\
e_0 b_3+3\tau b_4&=\omega^k(\mu_2 a_1+3 f_1 a_3),\\
\sigma b_3+3 f_0 b_4&=\omega^k(\mu_2 a_2+3 f_1 a_4).
\end{aligned}
\label{eq:reduc1-fourier}
\end{equation}
These equations determine $(b_1,b_2)$ and $(b_3,b_4)$ in terms of $\bm a$, since the determinant of each $2\times2$ block, $D=3(e_0 f_0-\sigma\tau)$, is nonzero under
condition (M1). Thus,
\begin{equation}
 b_i=\omega^k\, M_i({\bm a})\qquad(i=1,2,3,4),
\label{eq:reduc1-sol}
\end{equation}
where
\begin{equation}
\begin{aligned}
 M_1\!&=\!\frac{f_0(e_1 a_1+3\mu_1 a_3)-\tau(e_1 a_2+3\mu_1 a_4)}{e_0 f_0-\sigma\tau},\\
 M_2\!&=\!\frac{e_0(e_1 a_2+3\mu_1 a_4)-\sigma(e_1 a_1+3\mu_1 a_3)}{3(e_0 f_0-\sigma\tau)},\\
 M_3\!&=\!\frac{f_0(\mu_2 a_1+3 f_1 a_3)-\tau(\mu_2 a_2+3 f_1 a_4)}{e_0 f_0-\sigma\tau},\\
 M_4\!&=\!\frac{e_0(\mu_2 a_2+3 f_1 a_4)-\sigma(\mu_2 a_1+3 f_1 a_3)}{3(e_0 f_0-\sigma\tau)}.
\end{aligned}
\label{eq:M-formulas}
\end{equation}

The analogous conditions from the four length-$3$ basis strings
$E_{-1}E_{+1}E_0$, $E_{-1}E_{+1}F_0$, $F_{-1}E_{+1}E_0$, and
$F_{-1}E_{+1}F_0$ take the same form as~\eqref{eq:reduc1-fourier} with $\bm b$ replaced by $\bm c$.
Hence,
\begin{equation}
 c_i=\omega^k\, M_i({\bm a})\qquad(i=1,2,3,4),
\label{eq:reduc2-sol}
\end{equation}
and in particular \eqref{eq:reduc1-sol} and~\eqref{eq:reduc2-sol} give
$ c_i= b_i$.

Similarly, the vanishing of the coefficients of the four length-$3$ basis strings
\begin{equation}
E_{+1}E_0E_{-1},\quad F_{+1}F_0F_{-1},\quad
F_{+1}F_0E_{-1},\quad F_{+1}E_0F_{-1}
\label{eq:opt-triple}
\end{equation}
yields, respectively,
\begin{align}
0&=e_1 b_1(\omega^k-1)+\omega^k\mu_1 b_3-\mu_2 b_2,
\label{eq:reduc3-eq1}\\
0&=f_1(\omega^k b_2- b_3)+(\omega^k\mu_2-\mu_1) b_4,
\label{eq:reduc3-eq3}\\
0&=\omega^k f_1 b_1+\mu_2 b_3(\omega^k-1)-e_1 b_4,
\label{eq:reduc3-eq4}\\
0&=\omega^k\mu_2 b_2-\mu_1 b_3+f_1 b_4(\omega^k-1).
\label{eq:reduc3-eq5}
\end{align}

These four equations depend only on ${\bm b}=( b_1, b_2, b_3, b_4)$
and do not involve ${\bm a}$. We thus obtain a $4\times4$
constraint matrix $\bm A^{(B)}(\omega)$ acting on ${\bm b}$ (hereafter we
write $\omega^k$ simply as $\omega$):
\begin{equation}
\bm A^{(B)}(\omega)=\begin{pmatrix}
e_1(\omega-1) & -\mu_2 & \omega\mu_1 & 0\\
0 & \omega f_1 & -f_1 & \omega\mu_2-\mu_1\\
\omega f_1 & 0 & \mu_2(\omega-1) & -e_1\\
0 & \omega\mu_2 & -\mu_1 & f_1(\omega-1)
\end{pmatrix}.
\label{eq:AB-omega}
\end{equation}
Its determinant factorizes as
\begin{equation}
\det\bm A^{(B)}(\omega)=-\omega(\omega-1)\cdot R(\omega),
\label{eq:det-AB}
\end{equation}
where $R$ is a quadratic polynomial in $\omega$:
\begin{equation}
R(\omega)=c_2\,\omega^2+c_1\,\omega+c_0,
\label{eq:R-explicit}
\end{equation}
\begin{align}
c_2&=-(f_1-\mu_2)(f_1+\mu_2)(e_1\mu_2-f_1\mu_1),\label{eq:R-c2}\\
c_1&=\mu_2\bigl(2e_1f_1^2-e_1\mu_1\mu_2-e_1\mu_2^2\notag\\
   &\qquad+f_1\mu_1^2-f_1\mu_1\mu_2\bigr),\label{eq:R-c1}\\
c_0&=e_1^2f_1\mu_1-e_1^2f_1\mu_2-e_1f_1^2\mu_2\notag\\
   &\quad+e_1\mu_1\mu_2^2-f_1^3\mu_2+f_1\mu_1^2\mu_2.\label{eq:R-c0}
\end{align}
Note that $R$ depends only on $(e_1,f_1,\mu_1,\mu_2)$ and not on
$(e_0,f_0,\sigma,\tau)$.

Substituting the area-weighted Motzkin parameters~\eqref{eq:t-params}
into~\eqref{eq:R-c2}--\eqref{eq:R-c0}, we find
\begin{equation}
R^{(t)}(\omega)=+\tfrac{t^4}{32}\,(\omega-3),
\label{eq:R-t}
\end{equation}
whose only root is $\omega=3$.

For
$k\in\{1,\ldots,N-1\}$, $\omega^k=e^{2\pi ik/N}$ lies on the unit circle and $\omega^k\neq1$, so
\[
\omega^k\ne 0,\quad \omega^k\ne 1,\quad \omega^k\ne 3.
\]
Hence the three factors $-\omega$, $(\omega-1)$, and
$R^{(t)}(\omega)\propto(\omega-3)$ in~\eqref{eq:det-AB} are all
nonzero, giving
\begin{equation}
\det\bm A^{(B)}(\omega^k)\ne 0\qquad(\forall\,k\in\{1,\ldots,N-1\},\ \forall t\in\mathbb R\setminus\{0\}),
\end{equation}
so $\bm A^{(B)}(\omega^k)$ has full rank $4$ and ${\bm b}^{(k)}=0$.
Moreover, the linear map ${\bm a}\mapsto{\bm b}$ of Eq.~\eqref{eq:reduc1-sol} is invertible under conditions (M1) and (M3) [its determinant is proportional to powers of $e_0f_0-\sigma\tau$ and $e_1f_1-\mu_1\mu_2$], so
${\bm a}^{(k)}=0$, and Eq.~\eqref{eq:reduc2-sol} then gives
${\bm c}^{(k)}=0$. Therefore
\begin{equation}
 q_{ab}^{(k)}=0\qquad(\forall k\in\{1,\ldots,N-1\}),
\label{eq:trivial-k}
\end{equation}
and the nontrivial coefficients are restricted to the $k=0$ mode.

For $k=0$, we set $\omega^k=1$
in~\eqref{eq:reduc1-fourier} and \eqref{eq:reduc3-eq1}--\eqref{eq:reduc3-eq5}.
Equation~\eqref{eq:reduc1-fourier} becomes
\begin{equation}
\begin{aligned}
e_0 b_1+3\tau b_2&=e_1 a_1+3\mu_1 a_3,\\
\sigma b_1+3 f_0 b_2&=e_1 a_2+3\mu_1 a_4,\\
e_0 b_3+3\tau b_4&=\mu_2 a_1+3 f_1 a_3,\\
\sigma b_3+3 f_0 b_4&=\mu_2 a_2+3 f_1 a_4,
\end{aligned}
\label{eq:reduc1-omega1}
\end{equation}
which, under condition (M1), gives
\begin{equation}
\begin{aligned}
b_1&=\frac{f_0(e_1 a_1+3\mu_1 a_3)-\tau(e_1 a_2+3\mu_1 a_4)}{e_0 f_0-\sigma\tau},\\
b_2&=\frac{e_0(e_1 a_2+3\mu_1 a_4)-\sigma(e_1 a_1+3\mu_1 a_3)}{3(e_0 f_0-\sigma\tau)},\\
b_3&=\frac{f_0(\mu_2 a_1+3 f_1 a_3)-\tau(\mu_2 a_2+3 f_1 a_4)}{e_0 f_0-\sigma\tau},\\
b_4&=\frac{e_0(\mu_2 a_2+3 f_1 a_4)-\sigma(\mu_2 a_1+3 f_1 a_3)}{3(e_0 f_0-\sigma\tau)},
\end{aligned}
\label{eq:b-omega1}
\end{equation}
and  $c_i=b_i$.

Also, substituting $\omega^k=1$
into~\eqref{eq:reduc3-eq1}, \eqref{eq:reduc3-eq3}, and
\eqref{eq:reduc3-eq4} [Eq.~\eqref{eq:reduc3-eq5} at $\omega^k=1$ coincides with Eq.~\eqref{eq:reduc3-eq1} up to an overall sign, so it carries no additional information]:
\begin{itemize}\setlength{\itemsep}{2pt}
\item \eqref{eq:reduc3-eq1}: $\mu_1 b_3-\mu_2 b_2=0$.
\item \eqref{eq:reduc3-eq3}: $f_1(b_2-b_3)+(\mu_2-\mu_1) b_4=0$.
\item \eqref{eq:reduc3-eq4}: $f_1 b_1-e_1 b_4=0$.
\end{itemize}
Viewing these three as a $3\times4$ constraint matrix
$\bm A^{(B)}_0\in\mathbb C^{3\times 4}$ on $\bm b$,
\begin{equation}
\bm A^{(B)}_0=
\begin{pmatrix}
0 & -\mu_2 & \mu_1 & 0\\
0 & f_1 & -f_1 & \mu_2-\mu_1\\
f_1 & 0 & 0 & -e_1
\end{pmatrix},
\label{eq:AB-matrix}
\end{equation}
the $3\times 3$ minors (dropping each column) are
\begin{equation}
\begin{aligned}
\text{drop }b_1:&\;\det=e_1 f_1(\mu_1-\mu_2),\\
\text{drop }b_2:&\;\det=-f_1\mu_1(\mu_1-\mu_2),\\
\text{drop }b_3:&\;\det=f_1\mu_2(\mu_1-\mu_2),\\
\text{drop }b_4:&\;\det=-f_1^2(\mu_1-\mu_2).
\end{aligned}
\label{eq:AB-minors}
\end{equation}
For the area-weighted Motzkin chain, $f_1=-t/4\ne 0$ ($\forall t\ne 0$)
and $\mu_1=-\mu_2\neq0$, so all four minors are nonzero. In particular, the matrix~\eqref{eq:AB-matrix} has rank $3$, and its kernel is one dimensional. One checks directly that the vector
\begin{equation}
(b_1,b_2,b_3,b_4)=q_{k=2}\cdot(e_1,\mu_1,\mu_2,f_1),
\label{eq:b-sol}
\end{equation}
with $q_{k=2}$ an arbitrary constant,
annihilates all three rows of~\eqref{eq:AB-matrix} and therefore spans the kernel.
The corresponding $\bm a$ is the unique solution of~\eqref{eq:reduc1-omega1}; substituting
\begin{equation}
(a_1,a_2,a_3,a_4)=q_{k=2}\cdot(e_0,\sigma,\tau,f_0)
\label{eq:a-sol}
\end{equation}
into~\eqref{eq:b-omega1} indeed reproduces~\eqref{eq:b-sol}, e.g.,
\begin{align*}
b_1
=\frac{q_{k=2}\,e_1(e_0 f_0-\sigma\tau)}{e_0 f_0-\sigma\tau}=q_{k=2}\,e_1
\end{align*}
(the $\mu_1$ terms cancel), and similarly for $b_2,b_3,b_4$; moreover $c_i=b_i$. Comparing~\eqref{eq:a-sol} and \eqref{eq:b-sol} with the couplings in~\eqref{eq:H-param}, we conclude that, for all twelve $W=0$ components,
\begin{equation}
 q_{ab}^{(0)}=q_{k=2}\cdot c_{ab},
\label{eq:Q2-sol}
\end{equation}
where $c_{ab}$ denotes the coefficient of $O_a\otimes O_b$ in $\hloc$.

By~\eqref{eq:trivial-k} and~\eqref{eq:Q2-sol}, the coefficients
$q_{ab,i}=q_{k=2}c_{ab}$ are independent of $i$, and therefore
\begin{equation}
\hat X^{(2)}=q_{k=2}\sum_{i=1}^{N}\sum_{a,b}c_{ab}\,O_{a,i}\otimes O_{b,i+1}=q_{k=2}\,\sum_i\hloc_i.
\label{eq:X2-final}
\end{equation}
As in Sec.~\ref{sec:k2}, it follows that every $2$-local conserved quantity of $H^{(t)}$ is a linear combination of $H^{(t)}$, $1$-local conserved quantities, and the identity, which proves (B).

\subsection{Absence of $3$-local quantities with $\len([Q,H^{(t)}])\le2$}
\label{ssec:motzkin-step2}
Having verified conditions (A) and (B), we can apply Theorems~\ref{thm:hokkyo} and \ref{thm:hokkyo2} to the area-weighted periodic Motzkin chain: to rule out $k$-local conserved quantities for all $3\le k\le N/2$, it suffices to show that no $3$-local quantity $Q$ satisfies $\len([Q,H^{(t)}])\le2$.
As in Sec.~\ref{sec:reduction}, the strictly $3$-local part of such a $Q$ is fixed to be
\begin{equation}
\hat X^{(3)}=q\sum_i\bigl[\hloc_i,\hloc_{i+1}\bigr]
\label{eq:motzkin-X3}
\end{equation}
with a site-independent constant $q$ [the analogue of $q_{k=3}$; note that, because $\hloc$ now contains the cross terms, the length-$3$ strings of $\hat X^{(3)}$ are of the form $O_a[O_b,O_c]O_d$ with $c_{ab}c_{cd}\neq0$].
The strictly $2$-local part of $Q$ also enters the length-$3$ conditions; among its coefficients we will need only that of $F_{+2,i}F_{-2,i+1}$, which we denote by $\tilde q_{F_{+2},F_{-2}}$ [it is site independent because, in the relevant weight sector, the $k\neq0$ Fourier modes obey the homogeneous full-rank system of the previous subsection, while the source term $\hat X^{(3)}$ is translation invariant].
Below we derive two independent linear conditions on $(q,\tilde q_{F_{+2},F_{-2}})$ by demanding that the coefficients of the two length-$3$ basis strings $\vb B = F_{+2}E_{-1}E_{-1}$ and $\vb B=E_{+1}E_{+1}F_{-2}$ in $[Q,H^{(t)}]$ vanish, and show that they force $q=0$.

In the column expressions below, each length-$3$ string of $\hat X^{(3)}$ shares \emph{two} sites with the $2$-local interaction term of $H^{(t)}$.  As noted in Ref.~\cite{Hokkyo2024}, the commutator of two operators that share more than one site is in general not a single basis string but a linear combination [its evaluation uses the anticommutators of Table~\ref{tab:anticomm} as well]; we therefore mark such a commutator with a \emph{double horizontal line} in the column expression, retaining below the line only the string $\vb B$ of interest, with its coefficient shown on the left.  The polynomial displayed after each arrow is the total contribution of that placement to the coefficient of $\vb B$, per unit $q$; it includes the coefficient with which the first-row string appears in $\sum_i[\hloc_i,\hloc_{i+1}]$ and the coupling of the Hamiltonian term in the second row.
\subsubsection{\texorpdfstring{Term $\vb B=F_{+2}E_{-1}E_{-1}$}{Term B = F+2 E-1 E-1}}

The contributions from length-$3$ strings of $\hat X^{(3)}$ consist of the following eight
placements (the other conceivable placements contribute no $\vb B$ component):
\begin{widetext}
\begin{align*}
\begin{array}{cccc}
 & E_{+1} & E_{0} & E_{-1} \\
 & E_{+1} & E_{-1} & \\\hline\hline
-1 & F_{+2} & E_{-1} & E_{-1}
\end{array} \!\!\! &\Longrightarrow\, e_1^3+e_1\mu_1\mu_2,&
\begin{array}{cccc}
 & E_{+1} & F_{0} & E_{-1} \\
 & E_{+1} & F_{-1} & \\\hline\hline
-3 & F_{+2} & E_{-1} & E_{-1}
\end{array} \!\!\! &\Longrightarrow\, 3e_1\mu_1^2+3e_1\mu_1\mu_2,\\
\begin{array}{cccc}
 & F_{+1} & E_{0} & E_{-1} \\
 & F_{+1} & E_{-1} & \\\hline\hline
1 & F_{+2} & E_{-1} & E_{-1}
\end{array} \!\!\! &\Longrightarrow\, -e_1\mu_2^2-f_1\mu_2^2,&
\begin{array}{cccc}
 & F_{+1} & F_{0} & E_{-1} \\
 & F_{+1} & F_{-1} & \\\hline\hline
3 & F_{+2} & E_{-1} & E_{-1}
\end{array} \!\!\! &\Longrightarrow\, -3e_1f_1^2-3f_1\mu_2^2,\\
\begin{array}{cccc}
 & E_{+1} & E_{0} & E_{-1} \\
 & F_{+1} & F_{-1} & \\\hline\hline
-1 & F_{+2} & E_{-1} & E_{-1}
\end{array} \!\!\! &\Longrightarrow\, e_1^2f_1+f_1\mu_1\mu_2,&
\begin{array}{cccc}
 & E_{+1} & F_{0} & E_{-1} \\
 & F_{+1} & E_{-1} & \\\hline\hline
1 & F_{+2} & E_{-1} & E_{-1}
\end{array} \!\!\! &\Longrightarrow\, -e_1\mu_1\mu_2-e_1\mu_2^2,\\
\begin{array}{cccc}
 & F_{+1} & E_{0} & E_{-1} \\
 & E_{+1} & F_{-1} & \\\hline\hline
1 & F_{+2} & E_{-1} & E_{-1}
\end{array} \!\!\! &\Longrightarrow\, -e_1\mu_1\mu_2-f_1\mu_1\mu_2,&
\begin{array}{cccc}
 & F_{+1} & F_{0} & E_{-1} \\
 & E_{+1} & E_{-1} & \\\hline\hline
-1 & F_{+2} & E_{-1} & E_{-1}
\end{array} \!\!\! &\Longrightarrow\, e_1\mu_2^2+e_1^2f_1.
\end{align*}
\end{widetext}
Summing up all eight patterns,
\begin{align}
\mathcal M_2 :=&\; e_1^3+2e_1^2f_1-3e_1f_1^2+3e_1\mu_1^2\notag\\
&+2e_1\mu_1\mu_2-e_1\mu_2^2-4f_1\mu_2^2.
\label{eq:M2-full}
\end{align}
The contribution from length-2 operators is the single pattern
\begin{equation*}
\begin{array}{cccc}
 & F_{+2} & F_{-2} & \\
 & & F_{+1} & E_{-1} \\\hline
1 & F_{+2} & E_{-1} & E_{-1}
\end{array}\!\!\! \;\Longrightarrow\; \mu_2\,\tilde q_{F_{+2},F_{-2}}.
\end{equation*}
Therefore
\begin{equation}
0 = \mathcal M_2\,q + \mu_2\,\tilde q_{F_{+2},F_{-2}}.
\label{eq:step2-FEE}
\end{equation}

\subsubsection{\texorpdfstring{Term $\vb B=E_{+1}E_{+1}F_{-2}$}{Term B = E+1 E+1 F-2}}

Symmetrically, the contributions from length-$3$ strings of $\hat X^{(3)}$ consist of the following eight
placements:
\begin{widetext}
\begin{align*}
\begin{array}{cccc}
 & E_{+1} & E_{0} & E_{-1} \\
 & & E_{+1} & E_{-1} \\\hline\hline
1 & E_{+1} & E_{+1} & F_{-2}
\end{array} \!\!\! &\Longrightarrow\, -e_1^3-e_1\mu_1\mu_2,&
\begin{array}{cccc}
 & E_{+1} & E_{0} & F_{-1} \\
 & & E_{+1} & F_{-1} \\\hline\hline
-1 & E_{+1} & E_{+1} & F_{-2}
\end{array} \!\!\! &\Longrightarrow\, e_1\mu_1^2+f_1\mu_1^2,\\
\begin{array}{cccc}
 & E_{+1} & E_{0} & E_{-1} \\
 & & F_{+1} & F_{-1} \\\hline\hline
1 & E_{+1} & E_{+1} & F_{-2}
\end{array} \!\!\! &\Longrightarrow\, -e_1^2f_1-f_1\mu_1\mu_2,&
\begin{array}{cccc}
 & E_{+1} & E_{0} & F_{-1} \\
 & & F_{+1} & E_{-1} \\\hline\hline
-1 & E_{+1} & E_{+1} & F_{-2}
\end{array} \!\!\! &\Longrightarrow\, e_1\mu_1\mu_2+f_1\mu_1\mu_2,\\
\begin{array}{cccc}
 & E_{+1} & F_{0} & E_{-1} \\
 & & E_{+1} & F_{-1} \\\hline\hline
-1 & E_{+1} & E_{+1} & F_{-2}
\end{array} \!\!\! &\Longrightarrow\, e_1\mu_1^2+e_1\mu_1\mu_2,&
\begin{array}{cccc}
 & E_{+1} & F_{0} & F_{-1} \\
 & & E_{+1} & E_{-1} \\\hline\hline
1 & E_{+1} & E_{+1} & F_{-2}
\end{array} \!\!\! &\Longrightarrow\, -e_1^2f_1-e_1\mu_1^2,\\
\begin{array}{cccc}
 & E_{+1} & F_{0} & E_{-1} \\
 & & F_{+1} & E_{-1} \\\hline\hline
3 & E_{+1} & E_{+1} & F_{-2}
\end{array} \!\!\! &\Longrightarrow\, -3e_1\mu_1\mu_2-3e_1\mu_2^2,&
\begin{array}{cccc}
 & E_{+1} & F_{0} & F_{-1} \\
 & & F_{+1} & F_{-1} \\\hline\hline
-3 & E_{+1} & E_{+1} & F_{-2}
\end{array} \!\!\! &\Longrightarrow\, 3e_1f_1^2+3f_1\mu_1^2.
\end{align*}
\end{widetext}
Summing up all eight patterns,
\begin{align}
\mathcal M_1 :=&\; -e_1^3-2e_1^2f_1+3e_1f_1^2+e_1\mu_1^2\notag\\
&-2e_1\mu_1\mu_2-3e_1\mu_2^2+4f_1\mu_1^2.
\label{eq:M1-full}
\end{align}
The length-2 contribution is the single pattern
\begin{equation*}
\begin{array}{cccc}
 & & F_{+2} & F_{-2} \\
 & E_{+1} & F_{-1} & \\\hline
-1 & E_{+1} & E_{+1} & F_{-2}
\end{array}\!\!\! \;\Longrightarrow\; -\mu_1\,\tilde q_{F_{+2},F_{-2}}.
\end{equation*}
Therefore
\begin{equation}
0 = \mathcal M_1\,q - \mu_1\,\tilde q_{F_{+2},F_{-2}}.
\label{eq:step2-EEF}
\end{equation}

Viewing~\eqref{eq:step2-FEE} and~\eqref{eq:step2-EEF} as a system in
$(q,\tilde q_{F_{+2},F_{-2}})$,
\begin{equation}
\begin{pmatrix}\mathcal M_1 & -\mu_1\\ \mathcal M_2 & \mu_2\end{pmatrix}
\begin{pmatrix}q\\ \tilde q_{F_{+2},F_{-2}}\end{pmatrix}=0,
\end{equation}
the determinant is
\begin{equation}
\mathcal D := \mu_2 \mathcal M_1+\mu_1 \mathcal M_2 = (\mu_1-\mu_2)\,P(e_1,f_1,\mu_1,\mu_2),
\label{eq:Delta-fact}
\end{equation}
where
\begin{align}
P =&\; e_1^3+2e_1^2f_1-3e_1f_1^2+3e_1\mu_1^2\notag\\
  &+6e_1\mu_1\mu_2+3e_1\mu_2^2+4f_1\mu_1\mu_2.
\label{eq:P-def}
\end{align}

Substituting $e_1=-\tfrac{3t}{4}, f_1=-\tfrac{t}{4},
\mu_1=\tfrac{t}{4}, \mu_2=-\tfrac{t}{4}$ from~\eqref{eq:t-params}:
\begin{equation*}
\begin{aligned}
&\mathcal M_1=\tfrac{t^3}{2},\quad \mathcal M_2=-\tfrac{t^3}{2},\quad \mu_1-\mu_2=\tfrac{t}{2},\\
&P=-\tfrac{t^3}{2},\quad \mathcal D=-\tfrac{t^4}{4}\;(\ne 0\ \forall t\ne 0).
\end{aligned}
\end{equation*}
Hence $\mathcal D\ne0$ for every
$t\in\mathbb R\setminus\{0\}$, which forces $q=\tilde q_{F_{+2},F_{-2}}=0$. Therefore no $3$-local quantity $Q$ satisfies $\len([Q,H^{(t)}])\le2$, and by Theorem~\ref{thm:hokkyo} the area-weighted periodic Motzkin chain has no $k$-local conserved quantity for any $3\le k\le N/2$.
We note that this result is consistent with the conjecture of Ref.~\cite{Motzkin2504} that the periodic Motzkin chain at $t=1$ possesses a large symmetry algebra $C_N=\mathrm{sp}_{2N}$ generated by raising and lowering operators $\Sigma^\pm$ commuting with the Hamiltonian, which are observed there to be essentially nonlocal: the absence of $k$-local conserved quantities established here forces any conserved quantity of length $\ge3$, including any realization of the conjectured $\mathrm{sp}_{2N}$ generators, to have $\len>N/2$.

\section{Degenerate ground states of the periodic Motzkin chain}
\label{app:motzkin-gs}

Appendix~\ref{app:motzkin} established the nonintegrability of the
area-weighted periodic Motzkin chain.  Here we return to the undeformed
($t=1$) chain and prove a conjecture on its ground-state structure posed in Ref.~\cite{Motzkin2504}, by
an explicit construction of its frustration-free ground
states.

Throughout this appendix we label the single-site basis by $|u\rangle=|{+}1\rangle$, $|f\rangle=|0\rangle$, and
$|d\rangle=|{-}1\rangle$ (the eigenstates of $E_0$ with eigenvalues
$+1,0,-1$), and refer to a product state
$|l_1\cdots l_N\rangle$ with $l_i\in\{u,f,d\}$ as a \emph{configuration}.

The local terms of the periodic Motzkin chain [Eq.~\eqref{eq:Ht-def} at $t=1$] read, in
this notation,
\begin{equation}
H_{i,i+1}=U_{i,i+1}+D_{i,i+1}+F_{i,i+1},
\label{eq:Hlocal-motzkin}
\end{equation}
\begin{equation}
\begin{aligned}
U&=\bigl(|uf\rangle-|fu\rangle\bigr)\bigl(\langle uf|-\langle fu|\bigr),\\
D&=\bigl(|df\rangle-|fd\rangle\bigr)\bigl(\langle df|-\langle fd|\bigr),\\
F&=\bigl(|ud\rangle-|ff\rangle\bigr)\bigl(\langle ud|-\langle ff|\bigr),
\end{aligned}
\label{eq:UDF-proj}
\end{equation}
and the full periodic Hamiltonian is
\begin{equation}
H^{\mathrm{PBC}}=\sum_{i=1}^{N-1} H_{i,i+1}+H_{N,1}.
\label{eq:Hpbc}
\end{equation}
Its ground-state structure is described by the following conjecture.

\begin{conjecture}[Pronko~\cite{Motzkin2504}]\label{conj:motzkin}
The $N$-site periodic Motzkin chain~\eqref{eq:Hpbc} is frustration-free, and its ground space is $(2N+1)$-fold degenerate.  The ground
states $|v_{n}\rangle$ are labeled by the eigenvalue
$n=0,\pm1,\pm2,\ldots,\pm N$ of $S^z=\sum_i E_{0,i}$, and
$|v_{n}\rangle$ is the equal-weight superposition of all length-$N$
lattice paths from $(0,0)$ to $(N,n)$ with steps $\Delta x=1$,
$\Delta y\in\{-1,0,+1\}$ (the height constraint $y\ge0$ of the original Motzkin paths being absent on
the ring).
\end{conjecture}

Because $H^{\mathrm{PBC}}$ is assembled from $S^z$-conserving
projectors, $[H^{\mathrm{PBC}},S^z]=0$, so the ground space splits
into sectors of definite $S^z$ and may be analyzed sector by sector.
Below we construct, for every $S^z=n$, a zero-energy ground state,
and show that each sector contains exactly one ground state.
This proves both the frustration-freeness and the exact $(2N+1)$-fold
degeneracy asserted in Conjecture~\ref{conj:motzkin}.

\subsection{Construction of the frustration-free ground states}
\label{ssec:motzkin-ff}

We encode the two local configurations appearing in each of the projected states of~\eqref{eq:UDF-proj} by the \emph{reversible
local moves}
\begin{equation}
\begin{aligned}
g_U:\ &(uf)\leftrightarrow(fu),\\
g_D:\ &(df)\leftrightarrow(fd),\\
g_F:\ &(ud)\leftrightarrow(ff),
\end{aligned}
\label{eq:moves}
\end{equation}
acting on any pair of adjacent sites (including the bond between sites $N$ and $1$); note that $(du)$ is
not moved by $g_F$.  These moves generate an equivalence relation on the set of configurations.
Let $\mathcal C$ be an equivalence class and set
$|v_{\mathcal C}\rangle=\sum_{c\in\mathcal C}|c\rangle$.  We claim that $|v_{\mathcal C}\rangle$ is a
zero-energy ground state of~\eqref{eq:Hpbc}.
Indeed, consider the projector $U_{i,i+1}$ acting on $|v_{\mathcal C}\rangle$.  The configurations in $\mathcal C$ carrying a $(uf)$ pair on the bond $(i,i+1)$ come in pairs related by $g_U$, so $|v_{\mathcal C}\rangle$ contains only the symmetric combinations
$|{\cdots}uf{\cdots}\rangle+|{\cdots}fu{\cdots}\rangle$, which are
orthogonal to the state $|uf\rangle-|fu\rangle$ projected onto by $U$;
the same holds for $D$ and $F$.  Hence $H_{i,i+1}|v_{\mathcal C}\rangle=0$ for
every bond, and since $H^{\mathrm{PBC}}\ge0$, $|v_{\mathcal C}\rangle$ is a zero-energy ground
state.
Conversely, every zero-energy state $|\psi\rangle$ is annihilated by each projector, and the condition $\bigl(\langle uf|-\langle fu|\bigr)_{i,i+1}|\psi\rangle=0$ (and its $D$-, $F$-analogues) states that the amplitudes of two configurations related by a single move are equal.  Consequently, the amplitude of $|\psi\rangle$ is constant on each equivalence class, and the states $|v_{\mathcal C}\rangle$ span the zero-energy space:  the number of zero-energy ground states equals the number of equivalence classes.

It therefore suffices to show that the moves~\eqref{eq:moves} connect
all configurations of a given total weight $S^z=n$, i.e., that each $S^z$ sector forms a single equivalence class.  We prove
that any configuration $(l_1\cdots l_N)$ of weight $n\ge0$ can be
reduced to $(u^{\,n}f^{\,N-n})$ [for $n<0$,
replace $u^{\,n}$ by $d^{\,|n|}$, exchanging the roles of $g_U$ and $g_D$ below]; reversibility of the moves then makes the whole
$S^z=n$ sector a single class, whose equal-weight superposition is the
ground state $|v_n\rangle$.  Here $u^{\,k}$ ($f^{\,k}$) denotes $k$ consecutive
copies of $u$ ($f$).

\subsubsection*{Step 1.}
Applying $g_U$ and $g_D$ in the directions $(fu)\to(uf)$ and
$(fd)\to(df)$ moves every $f$ to the right,
\begin{equation}
(l_1\cdots l_N)\;\to\;(w\,f^{\,N-m}),
\label{eq:ff-step1}
\end{equation}
where $w=l'_1\cdots l'_m\in\{u,d\}^m$ is the sequence obtained by
deleting all $f$'s from the original configuration.

\subsubsection*{Step 2.}
Every adjacent pair $(ud)$ in $w$ is eliminated by applying $g_F$ in the direction
$(ud)\to(ff)$ and moving the produced $f$'s to the
right as in Step~1.  Each application of $g_{F}$ removes one $u$ and one $d$, leaving
$S^z$ unchanged.  Repeating this as many times as possible, we arrive at a sequence with no adjacent $(ud)$ pair, which is necessarily of the form
\begin{equation}
\;\to\;(d^{\,m_1}u^{\,m_2}f^{\,N-m_1-m_2}),\qquad m_2-m_1=n .
\label{eq:ff-step2}
\end{equation}

\subsubsection*{Step 3.}
If $m_1\geq1$, move the $u$ block to the right end with $g_U$,
\begin{equation}
\;\to\;(d^{\,m_1}f^{\,N-m_1-m_2}u^{\,m_2}),
\label{eq:ff-step3a}
\end{equation}
so that site $N$ carries $u$ and site $1$ carries $d$.  Owing to the periodic boundary condition, the
bond $(N,1)$ then reads $(ud)$, and $g_F$ converts it to
$(ff)$:
\begin{equation}
\;\to\;(f\,d^{\,m_1-1}f^{\,N-m_1-m_2}u^{\,m_2-1}f) .
\label{eq:ff-step3b}
\end{equation}
Repeating this procedure (moving
the new $f$'s to the right as in Step~1) removes the shorter block; for $n\ge0$ the
$d$ block is exhausted after $m_1$ iterations:
\begin{equation}
\;\to\;(f^{\,N-n}u^{\,n}) .
\label{eq:ff-step3c}
\end{equation}

\subsubsection*{Step 4.}
Finally, $g_U$ moves the surviving $u$ block to the left end,
\begin{equation}
\;\to\;(u^{\,n}f^{\,N-n}),
\label{eq:ff-step4}
\end{equation}
completing the reduction.  Since every configuration of weight
$S^z=n$ reduces to the same configuration $(u^{\,n}f^{\,N-n})$, the $S^z=n$ sector forms a single equivalence class and hence contains exactly one ground state.  As the eigenvalue of $S^z$ takes the $2N+1$ integer values from $-N$ to $N$, the ground space is exactly $(2N+1)$-fold degenerate. This proves Conjecture~\ref{conj:motzkin}.

\bibliography{ffspin_abs}

\end{document}